\documentclass[11pt]{article}
\usepackage[T1]{fontenc}
\usepackage{lmodern}
\usepackage[margin=1in]{geometry}
\usepackage{amsmath,amssymb,amsthm,mathtools}
\usepackage{enumitem}
\usepackage{microtype}
\usepackage{needspace}
\usepackage{xcolor}
\usepackage{framed}
\usepackage{tikz}
\usetikzlibrary{arrows.meta,decorations.pathreplacing}
\newcounter{algorithm}
\usepackage[colorlinks=true,linkcolor=blue!55!black,citecolor=blue!55!black,urlcolor=blue!55!black]{hyperref}
\usepackage[nameinlink,noabbrev,capitalise]{cleveref}
\crefname{algorithm}{Algorithm}{Algorithms}
\setlist[enumerate]{topsep=5pt,itemsep=3pt}
\setlist[itemize]{topsep=5pt,itemsep=3pt}
\numberwithin{equation}{section}
\newtheorem{theorem}{Theorem}[section]
\newtheorem{lemma}[theorem]{Lemma}
\newtheorem{corollary}[theorem]{Corollary}

\theoremstyle{definition}
\newtheorem{definition}[theorem]{Definition}
\theoremstyle{remark}
\newtheorem{remark}[theorem]{Remark}
\DeclareMathOperator{\cw}{cw}
\DeclareMathOperator{\pw}{pw}
\DeclareMathOperator{\td}{td}
\DeclareMathOperator{\tw}{tw}
\DeclareMathOperator{\VS}{VS}
\DeclareMathOperator{\MS}{MS}
\DeclareMathOperator{\MRS}{MRS}
\DeclareMathOperator{\LP}{LP}
\DeclareMathOperator{\DLP}{DLP}
\DeclareMathOperator{\LPVS}{LP_{VS}}
\DeclareMathOperator{\Vol}{Vol}

\DeclareMathOperator{\cc}{cc}
\newcommand{\R}{\mathbb{R}}
\newcommand{\eps}{\varepsilon}
\newcommand{\set}[1]{\left\{#1\right\}}
\newcommand{\abs}[1]{\left|#1\right|}

\title{Greedy Region Growing: Improved Approximations for Cutwidth, Pathwidth, Flow-Expander Decomposition, and More}
\author{Aditya Anand\thanks{Email: \href{mailto:adanand@umich.edu}{\texttt{adanand@umich.edu}}. Part of this work was done while the author was a Ph.D. student at the University of Michigan.}\\Max Planck Institute for Informatics}
\date{}

\begin{document}
\maketitle

\begin{abstract}
\normalsize
We give a new rounding technique for spreading-metric linear programs, called \emph{Greedy Region Growing}, which adapts the classical technique of region growing to min-max graph layout problems and flow-expander decomposition. Using this technique, we obtain state-of-the-art approximation algorithms for several fundamental problems.
\begin{enumerate}[leftmargin=*,label=\arabic*.]
\item An $O(\log n)$-approximation for \textsc{Cutwidth} and \textsc{Pathwidth}, improving upon the previous bound of $O(\log n\exp(O(\sqrt{\log\log n})))$ due to Bansal, Katzelnick, and Schwartz (FOCS 2024). This matches the $\Omega(\log n)$ integrality gap of the standard LP relaxations.

\item An $O(\log n)$-approximation for their directed analogues, \textsc{Minimum Storage} and \textsc{Minimum Register Sufficiency} on directed acyclic graphs, improving the previous bound of $O(\log^{3/2} n)$.
\item A $\phi$-flow-expander decomposition that cuts at most an $O(\phi\log n\log\log n)$ fraction of the edges. This improves the $O(\phi\log n\exp(O(\sqrt{\log\log n})))$ bound of Bansal, Jambulapati, and Saranurak (ICALP 2026), and is within an $O(\log\log n)$ factor of the graph-theoretic lower bound.
\end{enumerate}
The vertex version of our algorithm also gives an $O(\log n)$-approximation for \textsc{Treedepth}. In particular, our results answer several open problems of~\cite{BKS24}. Technically, we also solve the conceptual question raised by~\cite{BKS24}: our technique shows how to adapt Seymour recursion to min-max objectives, and constructs $O(1)$-weak local partitions in spreading metrics.
\end{abstract}

\clearpage
\setcounter{tocdepth}{1}
\tableofcontents
\clearpage

\section{Introduction}\label{sec:intro}
Graph partitioning is a central topic in approximation algorithms. Fundamental problems such as \textsc{Sparsest Cut}~\cite{LR99,LLR95,AR98,ARV09,ALN08}, \textsc{Multiway Cut}~\cite{DJPSY94,CKR00,SV14}, \textsc{Multicut}~\cite{GVY96}, \textsc{Feedback Vertex Set}~\cite{BG96,BBF99}, \textsc{Directed Feedback Arc Set}~\cite{Sey95,ENSS98}, \textsc{Treewidth} and \textsc{Pathwidth}~\cite{BGHK95,FHL05,BKS24} have led to a rich theory of approximation algorithms and metric embeddings. A common starting point for many of these problems is a linear program that assigns distances to vertex pairs and imposes \emph{spreading constraints}: sufficiently large vertex sets must be spread out in the resulting metric. The goal is then to round this metric into a cut, an ordering, or a decomposition. This approach was pioneered by the multicommodity-flow framework of Leighton and Rao~\cite{LR99} and the spreading-metric framework of Even, Naor, Rao, and Schieber~\cite{ENRS00}.

Leighton and Rao~\cite{LR99} introduced \emph{region growing} to round the linear programs for \textsc{Sparsest Cut} and \textsc{Balanced Separator}, obtaining an $O(\log n)$-approximation for the former with uniform demands and a bicriteria $O(\log n)$-approximation for the latter. In region growing, we repeatedly grow a ball in the metric until its boundary is cheap, and cut it out. For many problems, including \textsc{Minimum Linear Arrangement}, \textsc{Cutwidth} and expander decomposition~\cite{KVV04,ST11,SW19}, the natural approach is to decompose the graph recursively using balanced separators. This typically gives an $O(\log^2 n)$-approximation: since every separator is balanced, the pieces shrink by a constant factor at each level, so there are $O(\log n)$ levels of recursion, and at each level we lose an $O(\log n)$ factor in approximating \textsc{Balanced Separator}. The celebrated $O(\sqrt{\log n})$-approximation for balanced separators of Arora, Rao, and Vazirani~\cite{ARV09}, and its extension to vertex separators~\cite{FHL05}, improves the second logarithm, giving $O(\log^{3/2} n)$-approximations for several of these problems.

Seymour~\cite{Sey95}, in the context of feedback set problems in directed graphs, introduced a powerful refinement of region growing that trades the cost of a cut against the size of the piece it separates. Skipping many details, the idea is as follows. We may cut off pieces of any size (so the cuts need not be balanced), but how much we may pay depends on the size of the piece: a piece containing a constant fraction of the graph must be cut at nearly constant approximation cost, whereas a piece of constant size may be cut at cost $O(\log n)$. In general, we pay for a piece of size $s$ with a factor of $\log(n/s)$. The magic is that along a root-to-leaf path of the recursion tree, these factors telescope to a single $O(\log n)$: if the pieces along the path have sizes $n = s_1 \geq s_2 \geq \cdots \geq s_k = 1$, then the total factor is
\[
O\left(\log\frac{s_1}{s_2} + \log\frac{s_2}{s_3} + \cdots + \log\frac{s_{k-1}}{s_k}\right) = O\left(\log \frac{s_1}{s_k}\right) = O(\log n).
\]
The actual implementation works with LP volumes of metric balls rather than their sizes, and loses an additional $O(\log\log n)$ factor. Using this, Even, Naor, Rao, and Schieber~\cite{ENRS00} soon gave $O(\log n \log \log n)$-approximations for numerous graph partitioning problems, beating the $O(\log^2 n)$ barrier of balanced recursion. For min-sum objectives such as \textsc{Minimum Linear Arrangement}, Rao and Richa~\cite{RR05} improved the approximation to $O(\log n)$, and subsequent SDP-based techniques give $O(\sqrt{\log n}\log\log n)$-approximations~\cite{FL07,CHKR10}.

Min-max objectives such as \textsc{Cutwidth} and \textsc{Pathwidth} are much less well understood, and Seymour recursion does not directly apply to them. In \textsc{Cutwidth}, we must find an ordering of the vertices that minimizes the maximum number of edges leaving a prefix. Even if the cuts produced by Seymour recursion are individually cheap, their boundaries may all pile up on a single prefix. This is why, for a long time, the best approximation for these problems remained the $O(\log^{3/2} n)$ given by balanced separators. In a recent breakthrough, Bansal, Katzelnick, and Schwartz~\cite{BKS24} overcame this obstacle using a multiscale metric decomposition, obtaining $O(\log n\exp(O(\sqrt{\log\log n})))$-approximations for \textsc{Cutwidth} and \textsc{Pathwidth}. They explicitly point out why adapting Seymour recursion seems to fail, and they proposed \emph{weak local partitions} as a possible route to an $O(\log n\log\log n)$-approximation, leaving their construction open. Moreover, their multiscale argument does not directly extend to directed min-max layouts, where the metric is asymmetric.

The same obstacle appears in flow-expander decomposition. Previously, the best overhead was $O(\log^2 n)$ for flow expansion ($O(\log^{3/2} n)$ for cut expansion). Recently, Bansal, Jambulapati, and Saranurak~\cite{BJS26} broke this barrier by adapting the decomposition of~\cite{BKS24}, obtaining a flow-expander decomposition with overhead $O(\log n\exp(O(\sqrt{\log\log n})))$.

This raises two natural questions. The first is conceptual: \emph{can Seymour recursion be adapted to min-max graph layout problems and to expander decomposition?} The second asks how far one can go. For all the layout problems we consider, the spreading-metric LP has an integrality gap of $\Omega(\log n)$ on bounded-degree expanders (\cite{BKS24}, \cref{app:gaps}),\footnote{For \textsc{Cutwidth}, the gap compares $\cw(G)$ with $\max_{\emptyset\ne S\subseteq V(G)}\LP(G[S])/|S|$, and analogously for the other layout problems.} so $O(\log n)$ is the best one can hope for with this LP approach. The same gap holds for the stronger direct \textsc{Cutwidth} LP of~\cite{BKS24}, as we show in \cref{app:gaps}. \emph{Can we reach this bound, removing even the $O(\log\log n)$ loss of Seymour recursion?}

\paragraph{Our contribution.}
We answer both questions in the affirmative. We give two rounding algorithms, which together we call \emph{Greedy Region Growing}. Both repeatedly carve the heaviest remaining ball out of the metric, so that the cut costs are controlled not only piece by piece but on every prefix of the carving at once; this is exactly what min-max objectives require. They differ in how the radius of each cut is chosen: locally for each center, or globally for all centers at once.
\begin{enumerate}[leftmargin=*,label=(\arabic*)]
\item \emph{Greedy Seymour Carving} chooses the radius locally: it is a greedy variant of Seymour recursion, in which each center chooses its radius by Seymour's region growing. This answers the first question. With it, we obtain a flow-expander decomposition with $O(\log n\log\log n)$ overhead, and an $O(\log n\log\log n)$-approximation for \textsc{Cutwidth}. In fact, Greedy Seymour Carving gives $O(\log n\log\log n)$ bounds for all the problems in this paper; since Global Carving does better for the layout problems, we present it in full only for expander decomposition.
\item \emph{Global Carving} goes beyond Seymour recursion and answers the second question. It uses a single radius shared by all centers, which shrinks at a rate proportional to the volume of the heaviest remaining ball. The algorithm is simple and applies uniformly to all our layout problems, including directed layouts, where the metric is asymmetric. With it, we remove the $O(\log\log n)$ factor and obtain $O(\log n)$-approximations, matching the integrality gap, for \textsc{Cutwidth}, \textsc{Pathwidth}, \textsc{Treedepth}, \textsc{Minimum Storage} and \textsc{Minimum Register Sufficiency}. Flow-expander decomposition is the only remaining case, and we leave removing the $O(\log\log n)$ factor there open.
\end{enumerate}
We state these results below.

\subsection{Cutwidth and pathwidth}
\textsc{Cutwidth} and \textsc{Pathwidth} are classical graph layout problems whose approximation algorithms have developed alongside multicommodity-flow methods and balanced-separator techniques~\cite{LR99,BGHK95}. For an ordering $\pi=(v_1,\ldots,v_n)$ of an undirected graph, its cutwidth is the maximum number of edges with exactly one endpoint in a prefix of $\pi$. The cutwidth of the graph is the minimum over all orderings. \textsc{Pathwidth} has an equivalent ordering formulation: the cost of a prefix is the number of vertices in the prefix with a neighbor outside it. The minimum, over all orderings, of the maximum prefix cost is the vertex-separation number, which equals the pathwidth~\cite{Kin92}.

As a function of $n$, the best known approximation for both problems, due to Bansal, Katzelnick, and Schwartz~\cite{BKS24}, is $O(\beta(n)\log n)$, where $\beta(n)=\exp(O(\sqrt{\log\log n}))$. Assuming the Small Set Expansion Hypothesis, \textsc{Cutwidth} admits no constant-factor approximation~\cite{WAPL14}. We obtain $O(\log n)$-approximations for both problems.

\begin{theorem}[Cutwidth and pathwidth]\label{thm:main-layout}
There are deterministic polynomial-time algorithms that, given an $n$-vertex undirected graph $G$, compute an ordering of cutwidth at most
\[
O(\log n)\cw(G)
\]
and a path decomposition of width at most
\[
O(\log n)\pw(G),
\]
respectively.
\end{theorem}
This answers open questions of~\cite{BKS24}. The \textsc{Cutwidth} result is proved in \cref{sec:cutwidth}. The \textsc{Pathwidth} result follows from the vertex version of the carving algorithm in \cref{sec:vertex}.

\subsection{Directed min-max layouts}
In directed layout problems, the input is a directed acyclic graph and the output ordering must be topological. In \textsc{Minimum Storage}, the cost of an ordering is the maximum number of arcs crossing a prefix. In \textsc{Minimum Register Sufficiency}, the cost is the maximum number of vertices in a prefix with an out-neighbor outside the prefix. These classical problems model peak memory and register requirements in precedence-constrained computations. Approximation algorithms for \textsc{Minimum Register Sufficiency} date back to the early 1990s~\cite{KARR90}.

The best previous bound, $O(\log^{3/2} n)$, follows by combining balanced recursion with the directed separator algorithms of~\cite{ACMM05}. For \textsc{Minimum Register Sufficiency}, this combination with the framework of~\cite{RAK91} is noted explicitly in~\cite{WAPL14}; the edge-separator version gives the same bound for \textsc{Minimum Storage}. Bansal, Katzelnick, and Schwartz~\cite{BKS24} asked whether almost logarithmic approximations are possible; since the natural metric is asymmetric, the packing and covering properties used by their multiscale decomposition no longer hold.

\begin{theorem}[Directed layouts]\label{thm:main-directed}
There are polynomial-time $O(\log n)$-approximation algorithms for \textsc{Minimum Storage} and \textsc{Minimum Register Sufficiency} on directed acyclic graphs.
\end{theorem}
We prove the \textsc{Minimum Storage} result in \cref{sec:directed}. The reduction from \textsc{Minimum Register Sufficiency} to weighted \textsc{Minimum Storage} is given in \cref{app:register}.

\subsection{Treedepth}
Treedepth is a central parameter in structural graph theory, particularly in the study of sparse graph classes~\cite{CNP21}. An elimination forest of a graph is a rooted forest in which the endpoints of every edge are comparable in the ancestor order. The treedepth of the graph is the minimum height of such a forest. The vertex version of our algorithm gives the following result.

\begin{theorem}[Treedepth]\label{thm:main-treedepth}
There is a polynomial-time algorithm that, given an $n$-vertex graph $G$, constructs an elimination forest of height
\[
O(\log n)\td(G).
\]
\end{theorem}
Recursive balanced vertex separators~\cite{FHL05} give an $O(\log^{3/2}n)$-approximation. It can be shown that the separator construction in Section~5 of~\cite{BKS24} also implies an $O(\beta(n)\log n)$-approximation for \textsc{Treedepth}, although their paper states the result only for \textsc{Pathwidth}. We obtain an $O(\log n)$-approximation for treedepth as well.

There are also incomparable guarantees in terms of width parameters. Czerwi\'nski, Nadara, and Pilipczuk~\cite{CNP21} compute an elimination forest of height $O(kt\log^{3/2}t)$, where $k=\td(G)$ and $t=\tw(G)$; this is better when $t$ is small. On the hardness side, it is NP-hard to approximate \textsc{Treedepth} within a factor of $1.0003$~\cite{BNS25}.

\subsection{Flow-expander decomposition and the fine-grained flow-cut gap}
Expander decomposition is one of the most fundamental techniques in modern graph algorithms~\cite{KVV04,ST11,SW19}. We work with the stronger notion of flow expansion, as opposed to the traditional cut-based expansion. An undirected graph is a $\phi$-flow expander if every multicommodity demand respecting its vertex degrees can be routed with congestion at most $1/\phi$. A $\phi$-flow-expander decomposition is a set of edges whose deletion leaves only $\phi$-flow-expander components. Flow expansion certifies the routability of all degree-respecting demands, rather than only a lower bound on the capacity of every cut; in particular, a $\phi$-flow expander is also a $\phi$-cut expander.

A simple algorithm either certifies that the graph is a $\phi$-flow expander or finds an approximate sparsest cut and recurses, deleting $O(\phi m\log^2 n)$ edges. The best known bound, $O(\phi m\log n\exp(O(\sqrt{\log\log n})))$, is due to Bansal, Jambulapati, and Saranurak~\cite{BJS26}. We obtain the following.

\begin{theorem}[Flow-expander decomposition]\label{thm:main-expander}
There is a polynomial-time algorithm that, given an unweighted graph $G=(V,E)$ with $n$ vertices and $m$ edges and a rational parameter $0<\phi\le1$, returns $C\subseteq E$ such that every connected component of $G-C$ is a $\phi$-flow expander and
\[
|C|=O(\phi m\log n\log\log n).
\]
\end{theorem}
The theorem can be extended to edge-weighted graphs, with the total weights of $C$ and of $E$ in place of $|C|$ and $m$. The bound is within an $O(\log\log n)$ factor of the $\Omega(\log n)$ graph-theoretic lower bound on the overhead, which already holds for cut-expander decompositions (see~\cite{BJS26}).

The key ingredient, which gives the decomposition by simple recursion, is the following algorithmic \emph{fine-grained flow-cut gap}. Here the volume $\mathrm{vol}(S)$ of a set $S$ is the sum of the degrees of its vertices, and $\delta_G(S)$ is the set of edges with exactly one endpoint in $S$.

\begin{theorem}[Fine-grained flow-cut gap]\label{thm:fine-grained}
There is a polynomial-time algorithm that, given a connected graph $G=(V,E)$ and a parameter $\phi$, either certifies that $G$ is a $\phi$-flow expander, or outputs a set $S$ with $0<\mathrm{vol}(S)\le\mathrm{vol}(V)/2$ and
\[
|\delta_G(S)| \leq O(\phi\log \log n)\,\mathrm{vol}(S)\log \frac{\mathrm{vol}(V)}{\mathrm{vol}(S)}.
\]
\end{theorem}
This is the case $A=\deg_G$ of \cref{lem:size-cut}, which holds more generally for every node-weighting with all weights at least one. The classical flow-cut gap of Leighton and Rao~\cite{LR99} gives the same statement with the right-hand side replaced by $O(\phi\log n)\,\mathrm{vol}(S)$, and the work of~\cite{BJS26} implicitly gives a similar bound with $O(\log\log n)$ replaced by $\exp(O(\sqrt{\log\log n}))$. We conjecture that the $O(\log\log n)$ factor can also be removed from this fine-grained flow-cut gap, and leave this as an interesting open question.

\section{Technical Overview}\label{sec:overview}
We use \textsc{Cutwidth} to develop the main rounding idea of Greedy Region Growing: a greedy carving order that controls the cumulative boundary of every prefix. We first present the local version, Greedy Seymour Carving, and then Global Carving, which removes the remaining $O(\log\log n)$ factor. In \cref{sec:ov-other}, we briefly describe the additional ingredients needed for \textsc{Pathwidth} and \textsc{Treedepth}, directed layouts, and flow-expander decomposition.

\subsection{The LP and the recursive framework}\label{sec:ov-framework}
We begin with the LP relaxation for \textsc{Cutwidth} used in~\cite{BKS24}.\footnote{In fact, this is the standard relaxation for \textsc{Minimum Linear Arrangement}.} Recall that the goal is to find an ordering of the vertices that minimizes the maximum number of edges leaving a prefix of the ordering.
\begin{equation}\label{eq:cw-lp-overview}
\begin{aligned}
\LP(G)=\min\quad &\sum_{uv\in E(G)}d(u,v)\\
\text{s.t.}\quad &d(u,w)\le d(u,v)+d(v,w) &&u,v,w\in V(G),\\
&\sum_{v\in S}d(u,v)\ge\frac{|S|^2-1}{4} &&u\in V(G),\ S\subseteq V(G),\\
&d(u,v)=d(v,u) &&u,v\in V(G),\\
&d(u,u)=0,\quad d(u,v)\ge0 &&u,v\in V(G).
\end{aligned}
\end{equation}
For every induced subgraph $H$ of $G$, the quantity $\LP(H)/|V(H)|$ is a lower bound on the cutwidth of $G$. Indeed, let $d$ be the metric induced by an optimal cutwidth ordering of $H$, so that the vertices at positions $i$ and $j$ are at distance $|i-j|$. This metric is feasible, and its objective is the sum of the prefix cuts, which is at most $|V(H)|$ times the cutwidth of $H$, and hence at most $|V(H)|\cw(G)$.

Our algorithm is recursive. Suppose we find an ordered partition of $V(G)$ into pieces $X_1, X_2, \ldots, X_k$. This fixes the order between pieces: vertices of the first piece come first, then those of the second, and so on. We then recursively order each induced subgraph $G[X_1], G[X_2], \ldots, G[X_k]$.

\begin{figure}[t]
\centering
\begin{tikzpicture}[x=1cm,y=1cm,>=Latex,
  part/.style={draw=black!65,rounded corners=1pt,minimum width=1.72cm,minimum height=0.52cm,align=center,font=\footnotesize},
  before/.style={part,fill=black!12},
  path/.style={part,draw=blue!70!black,fill=blue!10,thick},
  after/.style={part,fill=white},
  edge/.style={draw=black!50,thin},
  chosen/.style={draw=blue!70!black,thick}]
  \node[path] (hzero) at (4.1,5.35) {$H_0=G$};
  \node[before,minimum width=2.25cm] (aone) at (1.15,4.18) {$X_1,\ldots,X_{i-1}$};
  \node[path,minimum width=2.0cm] (hone) at (4.1,4.18) {$H_1=G[X_i]$};
  \node[after,minimum width=2.25cm] (bone) at (7.05,4.18) {$X_{i+1},\ldots,X_k$};
  \draw[edge] (hzero.south) -- (aone.north);
  \draw[chosen] (hzero.south) -- (hone.north);
  \draw[edge] (hzero.south) -- (bone.north);
  \node[before,minimum width=2.25cm] (atwo) at (1.15,2.92) {earlier pieces};
  \node[path,minimum width=2.0cm] (htwo) at (4.1,2.92) {$H_2$};
  \node[after,minimum width=2.25cm] (btwo) at (7.05,2.92) {later pieces};
  \draw[edge] (hone.south) -- (atwo.north);
  \draw[chosen] (hone.south) -- (htwo.north);
  \draw[edge] (hone.south) -- (btwo.north);
  \node[font=\large] at (4.1,2.18) {$\vdots$};
  \node[path,minimum width=2.0cm] (hlast) at (4.1,1.48) {$H_T$};
  \draw[chosen,densely dotted] (htwo.south) -- (hlast.north);
  \node[before,minimum width=2.25cm] (alast) at (1.15,0.30) {earlier vertices};
  \node[path,minimum width=0.8cm] (v) at (4.1,0.30) {$v$};
  \node[after,minimum width=2.25cm] (blast) at (7.05,0.30) {later vertices};
  \draw[edge] (hlast.south) -- (alast.north);
  \draw[chosen] (hlast.south) -- (v.north);
  \draw[edge] (hlast.south) -- (blast.north);
  \draw[densely dashed,red!70!black,thick] (4.52,-0.09) -- (4.52,0.69);
  \node[font=\footnotesize,red!70!black,anchor=north] at (4.52,-0.16) {prefix $P$ ends at $v$};
  \node[font=\footnotesize,anchor=west] at (8.45,5.35) {contribution of each level};
  \node[font=\footnotesize,anchor=west] at (8.45,4.18) {$c_0=\alpha\cw(G)\log\dfrac{2e\,\LP(G)}{\LP(X_i)}$};
  \node[font=\footnotesize,anchor=west] at (8.45,2.92) {$c_1=\alpha\cw(G)\log\dfrac{2e\,\LP(H_1)}{\LP(H_2)}$};
  \node[font=\large,anchor=west] at (8.65,2.18) {$\vdots$};
  \node[font=\footnotesize,anchor=west] at (8.45,0.30) {$c_T=\alpha\cw(G)\log(2e|V(H_T)|)$};
  \node[font=\footnotesize,anchor=west] at (8.45,-0.78) {$|\delta_G(P)|\le\displaystyle\sum_{t=0}^T c_t=O(\alpha\log n)\cw(G)$};
\end{tikzpicture}
\caption{The cost of a prefix $P$ ending at a vertex $v\in X_i$. At each level, the shaded earlier pieces belong to $P$, the unshaded later pieces do not, and the blue piece contains $v$. At the top level, these are the pieces $X_1,\ldots,X_k$ of $G$. The cut receives one contribution $c_t$ from each level $t$; the weak local partition bound \eqref{eq:wlp-overview} makes their sum $O(\alpha\log n)\cw(G)$.}
\label{fig:recursive-overview}
\end{figure}
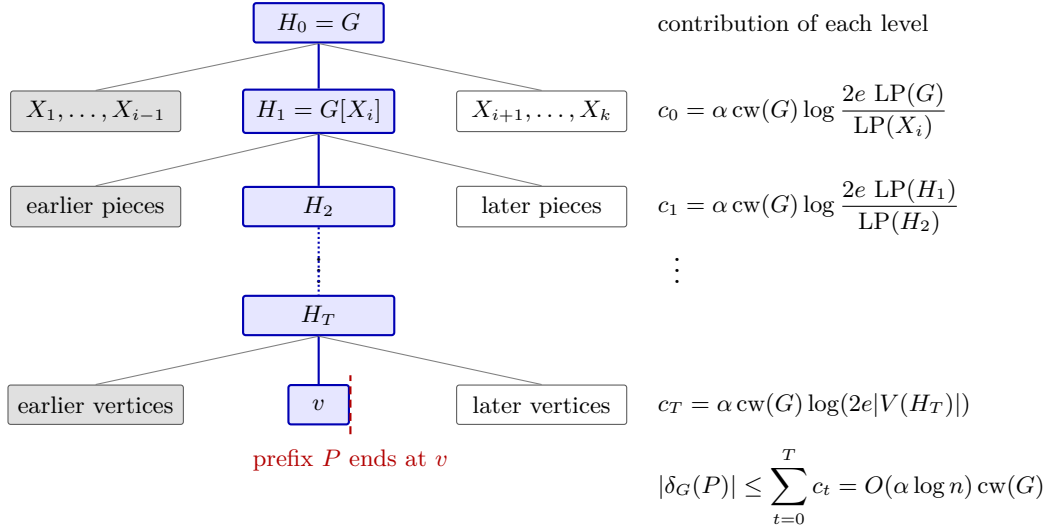

The guarantee we want from the partition is the following. For a set $X\subseteq V(G)$, write $\LP(X)$ for $\LP(G[X])$, and $\delta_G(X)$ for the set of edges with exactly one endpoint in $X$. Every piece should have at most $n/2$ vertices, and for every $i\le k$,
\begin{equation}\label{eq:wlp-overview}
\sum_{j\le i}|\delta_G(X_j)|
\;\le\;\alpha\,\frac{\LP(G)}{n}\,\log\frac{2e\,\LP(G)}{\LP(X_i)}.
\end{equation}
The left-hand side upper bounds the number of edges leaving the prefix $X_1\cup\cdots\cup X_i$.\footnote{Since $\LP(X_i)$ may be zero, the actual definition (\cref{def:wlp}) uses instead the \emph{volume} of $X_i$: the total LP length of the edges inside $X_i$, plus $\frac{\LP(G)}{n}$ per vertex. It is at least $\LP(X_i)$, so the actual condition implies \eqref{eq:wlp-overview}. The factor $2e$ only keeps the logarithm at least one.} Following the terminology of~\cite{BKS24}, we call such an ordered partition an \emph{$\alpha$-weak local partition}.\footnote{Our ordered definition differs slightly from that of~\cite{BKS24}, but our construction also satisfies theirs (\cref{app:bks-wlp}).}

Why does this suffice? Consider a prefix $P$ of the final ordering that ends at a vertex $v$ of the piece $X_i$, as illustrated in \cref{fig:recursive-overview}. Intuitively, the current level contributes about $\alpha\cw(G)\log\frac{\LP(G)}{\LP(X_i)}$ edges to its cut, while the recursion inside $X_i$ contributes about $\alpha\cw(G)\log\LP(X_i)$; the two logarithms add up to $\log\LP(G)=O(\log n)$. Concretely, let $G=H_0\supseteq H_1\supseteq\cdots\supseteq H_T$ be the recursive instances containing $v$: $H_1=G[X_i]$, and in general $H_{t+1}$ is induced by the piece of $H_t$ that contains $v$. Since $\LP(H_t)/|V(H_t)|\le\cw(G)$, applying \eqref{eq:wlp-overview} at each level shows that $P$ is crossed by at most
\[
\alpha\,\cw(G)\left(\sum_{t<T}\log\frac{2e\,\LP(H_t)}{\LP(H_{t+1})}+\log(2e|V(H_T)|)\right)
\;=\;O(\alpha\log n)\cw(G)
\]
edges, since the depth $T$ is $O(\log n)$ and the LP values lie between $\Omega(1)$ and $\mathrm{poly}(n)$. For the rest of the overview, we therefore focus on constructing $\alpha$-weak local partitions: first with $\alpha=O(\log\log n)$, using Greedy Seymour Carving (\cref{sec:ov-region,sec:ov-difficulty,sec:ov-greedy,sec:ov-payoff}), and then with $\alpha=O(1)$, using Global Carving (\cref{sec:ov-paced}).

\subsection{Balls, volumes and region growing}\label{sec:ov-region}
We first recall region growing, and give an alternative proof, different from the one in the classical literature, of Seymour's fine-grained version. This argument is the basic building block of Greedy Seymour Carving.

Fix an optimal solution $d$ of \eqref{eq:cw-lp-overview}. For a vertex $v$ and a radius $r\ge0$, the \emph{ball} $B_d(v,r)=\{u:d(v,u)<r\}$ is the set of vertices at distance less than $r$ from $v$. For every vertex $v$, the ball of radius $r=\frac{n}{10}$ contains at most $\frac n2$ vertices. Otherwise, it can be easily shown that the set $S=B_d(v,r)$ violates the spreading constraint of \eqref{eq:cw-lp-overview} with center $u=v$.

The \emph{volume} $\Vol_v(r)$ of the ball $B_d(v,r)$ is the total LP length of the edges inside it, plus $\frac{\LP(G)}{n}$ for every vertex in the ball. An edge on the boundary of the ball contributes the part of its length that lies inside. For example, if an edge $wz$ has $d(v,w)=0.2$ and $d(v,z)=0.3$, and $r=0.21$, then it contributes $0.21-0.2=0.01$. With this definition, the volume grows continuously as long as no new vertex enters the ball, and its derivative is exactly the cut-size of $B_d(v,r)$. When a new vertex enters the ball, the volume \emph{jumps} up: the term of the vertex is added, and its edges to vertices already in the ball now count with their full length. The term $\frac{\LP(G)}{n}$ per vertex keeps volumes positive: every ball has volume between $\frac{\LP(G)}{n}$ and $2\LP(G)$. We call
\[
F_v(r)=\log\frac{2e\,\LP(G)}{\Vol_v(r)}
\]
the \emph{deficit} of the ball. It satisfies $1\le F_v(r)\le\log(2en)$, and it is nonincreasing in $r$.

\paragraph{Classical region growing.} Given radii $a<b$, there is a radius $r\in[a,b]$ such that
\[
|\delta_G(B_d(v,r))| \leq O\left(\frac{\Vol_v(r)}{b-a} \log \frac{2\Vol_v(b)}{\Vol_v(a)}\right).
\]
To see this, divide $[a,b]$ into $k=\lceil\log_2(2\Vol_v(b)/\Vol_v(a))\rceil$ intervals of equal length. The volume cannot double across every one of them, since otherwise $\Vol_v(b)\geq 2^k\Vol_v(a)\geq 2\Vol_v(b)$. So some interval $[a',b']$ is \emph{stable}, meaning $\Vol_v(b') \leq 2\Vol_v(a')$. The average increase of $\Vol_v$ over this interval is at most $\frac{\Vol_v(a')}{b'-a'}$, so at some point the derivative is at most this quantity, since the jumps only increase $\Vol_v$. This gives a radius $r \in [a',b']$ with
\[
|\delta_G(B_d(v,r))| \leq \frac{\Vol_v(a')}{b'-a'} \leq \frac{k\,\Vol_v(r)}{b-a} = O\left(\frac{\Vol_v(r)}{b-a}\log\frac{2\Vol_v(b)}{\Vol_v(a)}\right).
\]

\paragraph{The fine-grained version.} The logarithm above involves $\Vol_v(a)$, the volume at the \emph{inner} radius, which may be as small as $\LP(G)/n$. So it can be $\Theta(\log n)$ even if the ball we eventually cut out is huge. We now show that there is a radius $r\in[a,b]$ satisfying
\[
|\delta_G(B_d(v,r))| \leq O\left(\frac{\log\log n}{b-a}\,\Vol_v(r)\,F_v(r)\right).
\]
The proof is the same argument applied one level up, to the deficit $F_v$ instead of the volume $\Vol_v$. Divide $[a,b]$ into $\kappa=\lceil\log_2\log(2en)\rceil+1$ intervals of equal length. Moving inward, the deficit cannot double across every one of them, since otherwise $F_v(a)>2^\kappa F_v(b)\ge2^\kappa>\log(2en)$. So some interval $[a'',b'']$ is stable for the deficit, meaning $F_v(a'')\le 2F_v(b'')$.\footnote{A related stability condition for logarithmic ball-size deficits appears in~\cite[Section~3.2]{BKS24}, in their multiscale decomposition.} Applying the classical bound inside $[a'',b'']$, whose length is $\Theta\big(\frac{b-a}{\log\log n}\big)$, gives a radius $r\in[a'',b'']$ with
\[
|\delta_G(B_d(v,r))| \leq O\left(\frac{\log\log n}{b-a}\,\Vol_v(r)\log\frac{2\Vol_v(b'')}{\Vol_v(a'')}\right).
\]
The point is that this logarithm is now small:
\[
\log\frac{\Vol_v(b'')}{\Vol_v(a'')} = F_v(a'')-F_v(b'') \le F_v(b'') \le F_v(r),
\]
by stability and then by monotonicity of $F_v$, since $r\le b''$. As $F_v(r)\geq 1$, the additive $\log 2$ is absorbed, which proves the claim. The logarithm is now controlled by the volume of the ball we actually cut out rather than by the volume at the inner radius; the price is the factor $O(\log\log n)$ from working in a shorter interval.

Taking $[a,b]=[\frac{n}{20},\frac{n}{10}]$, so that every ball we consider has at most $\frac n2$ vertices, we obtain the following Seymour-type guarantee (\cref{lem:one-step}): around any center $v$, we can find a ball $X=B_d(v,r)$ of volume $W=\Vol_v(r)$ with
\begin{equation}\label{eq:seymour-overview}
|\delta_G(X)| \leq O\Big(\log\log n\cdot\frac{W}{n}\,\log\frac{2e\,\LP(G)}{W}\Big).
\end{equation}
Since $W$ contains the length of every edge inside $X$, and restricting $d$ to $G[X]$ gives a feasible solution, we have $W\ge\LP(X)$. So the logarithm in \eqref{eq:seymour-overview} is at most $\log\frac{2e\,\LP(G)}{\LP(X)}$, the quantity in \eqref{eq:wlp-overview}.

\subsection{Why one Seymour step is not enough}\label{sec:ov-difficulty}
Our first goal is an $O(\log\log n)$-weak local partition: an ordered partition $X_1,\ldots,X_k$ of $V(G)$ into pieces of at most $n/2$ vertices such that, for every $i\le k$,
\[
\sum_{j\le i}|\delta_G(X_j)|
\;\le\;O(\log\log n)\,\frac{\LP(G)}{n}\,\log\frac{2e\,\LP(G)}{\LP(X_i)}.
\]
The most direct approach is to apply \eqref{eq:seymour-overview} repeatedly: cut out one ball, remove it, apply \eqref{eq:seymour-overview} again in the remaining graph, and output the pieces in the order in which they were removed. Each piece then pays for its own boundary with its own deficit. The problem is not obtaining one such piece; it is controlling the total cut-size of a prefix, since every piece $X_j$ with $j\le i$ contributes to the prefix ending at $X_i$.

Let $W_i$ be the ball volume of $X_i$ when it is carved, measured in the graph from which it is removed, and let $F_i=\log\frac{2e\,\LP(G)}{W_i}$ be its deficit. Two simple facts are useful. First, charging each edge when the first of its endpoints is removed shows that the volumes of all pieces sum to at most $2\LP(G)$: at most $\LP(G)$ from the vertex terms and at most $\LP(G)$ from the edges. Second, each edge between two pieces is cut, in the remaining graph, when the earlier of the two pieces is removed, and it appears at most twice on the left-hand side of \eqref{eq:wlp-overview}. Using \eqref{eq:seymour-overview}, we obtain the naive bound
\[
\begin{aligned}
\sum_{j\le i}|\delta_G(X_j)|
&\le O\Big(\frac{\log\log n}{n}\Big)\sum_{j\le i}W_j\,F_j\\
&\le O\Big(\frac{\log\log n}{n}\Big)\max_{j\le i}F_j\sum_{j\le i}W_j
= O\Big(\log\log n\cdot\frac{\LP(G)}{n}\Big)\max_{j\le i}F_j.
\end{aligned}
\]
If the deficits $F_i$ were roughly equal, we would be done. This is precisely where the approach breaks down: in the worst case, we may first remove many pieces of small volume, and then a final piece of very large volume. Then $F_k \ll F_j$ for all $j < k$, and the bound above is far from the bound $O(\log\log n\cdot\frac{\LP(G)}{n}F_k)$ that we need for the prefix ending at $X_k$.

This suggests extracting the pieces greedily: pieces of large volume should come first, and pieces of small volume later. A simple algorithm would pick the largest ball first. However, the sequence of cuts determines the graph in which we look for the next piece. Every removal shrinks the balls around the remaining centers, and the radius chosen by the argument of \cref{sec:ov-region} may change completely. So there is no fixed list of pieces that we could sort in advance.

\subsection{Greedy Seymour Carving: stable annuli}\label{sec:ov-greedy}
The observation that rescues the greedy approach is that we never need to know exactly which piece a center will cut. It suffices to know its deficit up to a factor of two, and a stable interval provides exactly this. From now on, volumes and deficits are always measured in the graph induced by the vertices that have not yet been removed. We divide the range of radii $[\frac{n}{20},\frac{n}{10}]$ into $\Theta(\log\log n)$ equal annuli. If an annulus $[a,b]$ is stable for $v$, that is, $F_v(a)\le 2F_v(b)$, then, since the deficit is nonincreasing in the radius, every $r\in[a,b]$ satisfies
\[
F_v(b)\;\le\;F_v(r)\;\le\;2F_v(b).
\]
So the volume at the outer radius $b$, which we call the \emph{key} of $v$, predicts the deficit of the piece that $v$ would cut up to a factor of two, whatever radius region growing picks inside the annulus. By \cref{sec:ov-region}, that radius cuts at most $O\big(\frac{\log\log n}{n}\Vol_v(r)F_v(r)\big)$ edges.

This suggests the following algorithm. Every vertex keeps a pointer to a current annulus, initially the outermost one. While vertices remain:
\begin{enumerate}[label=(\arabic*),leftmargin=*]
\item move the pointer of every remaining vertex inward until its annulus is stable;
\item pick the vertex with the largest key;
\item grow a region inside its annulus and remove the resulting ball. All other vertices keep their pointers.
\end{enumerate}
Two properties make this work.

\emph{Keys never increase.} Removing vertices only removes volume, and moving a pointer inward moves the outer radius inward. So every key can only decrease, and since we always select the maximum, the selected keys are nonincreasing. This is why the pointers must persist: if we reset them to the outermost annulus after each cut, as a naive repetition of the one-step argument would, keys could jump back up.

\emph{Pointers never run out.} The key fact is that removing vertices only shrinks balls, so it can only \emph{increase} deficits: once a deficit is large, it stays large. A pointer leaves the annulus $[a,b]$ only when $F_v(a)>2F_v(b)$. Then $a$ becomes the new outer radius, whose deficit is already more than $2F_v(b)$, and by the key fact it never drops below this value again. So the deficits that trigger successive moves more than double each time. Since deficits lie between $1$ and $\log(2en)$, each pointer moves only $O(\log\log n)$ times, so the annuli fixed at the start suffice for the entire carving.

\subsection{Putting it together}
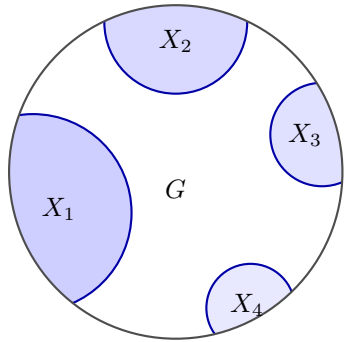
\label{sec:ov-payoff}
\begin{figure}[htbp]
\centering
\begin{minipage}[c]{0.49\linewidth}
\centering
\begin{tikzpicture}[x=0.9cm,y=0.9cm,>=Latex]
  \begin{scope}
    \clip (0,0) circle (2.45);
    \fill[blue!19] (-2.10,-0.60) circle (1.45);
    \fill[blue!15] (0,2.20) circle (1.05);
    \fill[blue!12] (2.15,0.55) circle (0.76);
    \fill[blue!9] (1.10,-2.00) circle (0.65);
    \draw[blue!65!black,thick] (-2.10,-0.60) circle (1.45);
    \draw[blue!65!black,thick] (0,2.20) circle (1.05);
    \draw[blue!65!black,thick] (2.15,0.55) circle (0.76);
    \draw[blue!65!black,thick] (1.10,-2.00) circle (0.65);
  \end{scope}
  \draw[black!70,thick] (0,0) circle (2.45);
  \node[font=\small] at (-1.72,-0.55) {$X_1$};
  \node[font=\small] at (0,1.93) {$X_2$};
  \node[font=\small] at (1.90,0.55) {$X_3$};
  \node[font=\small] at (1.05,-1.95) {$X_4$};
  \node[font=\small] at (0,-0.25) {$G$};
  \node[font=\footnotesize] at (0,-3.02) {carving order: $X_1\;\longrightarrow\;X_2\;\longrightarrow\;X_3\;\longrightarrow\;X_4$};
\end{tikzpicture}
\end{minipage}\hfill
\begin{minipage}[c]{0.48\linewidth}
\caption{Greedy carving of $G$. The selected keys are nonincreasing, so the pieces appear at roughly decreasing volume scales: $F_i\le2F_j$ for $i<j$. The sizes shown are schematic.}
\label{fig:greedy-volumes}
\end{minipage}
\end{figure}
\Needspace{7\baselineskip}
Let $X_1,X_2,\ldots$ be the pieces in the order of removal, and let $F_i$ be the deficit of $X_i$ when it is cut, as in \cref{sec:ov-difficulty}. By stability, $F_i$ is within a factor of two of the deficit of the $i$-th selected key, and these key deficits are nondecreasing. Hence
\[
F_i\le 2F_j\qquad\text{whenever } i<j.
\]
Up to a factor of two, the deficits of the pieces are nondecreasing: roughly, pieces of large volume come out first, which is exactly what the greedy approach wanted. The calculation of \cref{sec:ov-difficulty} now goes through, with the largest deficit replaced by $2F_i$:
\[
\sum_{j\le i}|\delta_G(X_j)|
\le O\Big(\frac{\log\log n}{n}\Big)\sum_{j\le i}W_j\,F_j
\le O\Big(\frac{\log\log n}{n}\Big)F_i\sum_{j\le i}W_j
= O\Big(\log\log n\cdot\frac{\LP(G)}{n}\,F_i\Big).
\]
Since $W_i\ge\LP(X_i)$, we have $F_i\le\log\frac{2e\,\LP(G)}{\LP(X_i)}$. So the carving order is an $O(\log\log n)$-weak local partition, and together with the recursion of \cref{sec:ov-framework}, this gives an $O(\log n\log\log n)$-approximation for \textsc{Cutwidth}.

\subsection{Removing the \texorpdfstring{$O(\log\log n)$}{O(log log n)} factor: Global Carving}\label{sec:ov-paced}
We now show how to remove the remaining factor $O(\log\log n)$ and construct an $O(1)$-weak local partition. We first explain why Greedy Seymour Carving cannot avoid this factor, and why a weak local partition does not actually require it. We then give the algorithm and show that it constructs an $O(1)$-weak local partition.

\paragraph{Where the factor $\log\log n$ comes from.}
Greedy Seymour Carving pays for every piece \emph{by itself}: the boundary of $X_i$ is charged to its own volume and its own deficit, as in \eqref{eq:seymour-overview}, and the prefix bound of \cref{sec:ov-payoff} is obtained by adding up these charges. The factor $\log\log n$ is the price of this local charging. To obtain \eqref{eq:seymour-overview} for a single ball, we must run region growing in an interval that is stable for the deficit. We do not know in advance where such an interval lies, so we split the range of radii into $\Theta(\log\log n)$ annuli, and region growing in an annulus of width $\Theta(n/\log\log n)$ loses this factor.

This loss is inherent in charging every piece to itself. Suppose that, as the radius decreases from $\frac n{10}$ to $\frac n{20}$, the deficit of a ball grows geometrically from a constant to $\Theta(\log n)$, say $F_v(r)=F_v(\frac n{10})\,D^{(n/10-r)/(n/20)}$ with $D=\Theta(\log n)$. Since $\Vol_v(r)=2e\,\LP(G)\,e^{-F_v(r)}$ and the derivative of $\Vol_v$ is the cut-size, we get, at every radius $r$ in this range,
\[
|\delta_G(B_d(v,r))|=-\Vol_v(r)\,F_v'(r)=\frac{20\log D}{n}\,\Vol_v(r)\,F_v(r)=\Theta\Big(\frac{\log\log n}{n}\Big)\Vol_v(r)\,F_v(r).
\]
So no radius cuts this ball at cost $O(\frac1n)\Vol_v(r)F_v(r)$. Such a ball is consistent with all the properties of volumes that we use, so an analysis based on these properties, in which every piece pays for its own boundary, must lose the factor $\log\log n$.

\paragraph{A global view.}
Let us look again at what a weak local partition asks for. As in \cref{sec:ov-payoff}, the deficit $F_i$ of a piece is at most $\log\frac{2e\,\LP(G)}{\LP(X_i)}$, so it suffices to show
\[
\sum_{j\le i}|\delta_G(X_j)|\;\le\;\alpha\,\frac{\LP(G)}{n}\,F_i\qquad\text{for every }i.
\]
The right-hand side depends only on the deficit $F_i$ of the \emph{current} piece. It is a budget for the whole prefix, and it does not ask $X_i$, or any earlier piece, to pay for its own boundary. Greedy Seymour Carving uses this budget in a particular way: it pays for every piece with its volume times its deficit, and then uses the fact that the volumes add up to $O(\LP(G))$. But the budget can also be used differently. Suppose that the deficits of the pieces were nondecreasing, $F_1\le F_2\le\cdots$. Then it would suffice that every cut is paid for by the \emph{increase} of the deficit from one piece to the next,
\[
|\delta_G(X_j)|\le O\Big(\frac{\LP(G)}{n}\Big)(F_{j+1}-F_j),
\]
since these increases add up to at most $F_i$; if, in addition, every cut costs at most $O(\frac{\LP(G)}{n})$, the cut of the current piece $X_i$ is absorbed as well, since $F_i\ge1$. Under this charging, a piece pays nothing for having a large deficit; it pays only for the growth of the deficit that it causes. The bound is then global: no single cut is compared with its own piece, only the carving as a whole is compared with the budget.

\paragraph{Why growing deficits pay for cuts.}
Why should the deficit grow when the cuts are expensive? For a single ball, the two are the same thing. Since $F_v=\log\frac{2e\,\LP(G)}{\Vol_v}$ and the derivative of $\Vol_v$ is the cut-size, away from the jumps of $\Vol_v$ we have
\[
|\delta_G(B_d(v,r))|=\Vol_v(r)\cdot\bigl(-F_v'(r)\bigr).
\]
So a ball has a large boundary exactly at the radii where its deficit grows quickly as the radius shrinks. This is what happens in the example above: the boundary is large at every radius, but the deficit also grows from a constant to $\Theta(\log n)$ across the range of radii. Averaging over a window makes this quantitative. Fix a window $(\rho-w,\rho)$, and let
\[
\Delta=F_v(\rho-w)-F_v(\rho),\qquad\text{so that}\qquad\Vol_v(\rho-w)=e^{-\Delta}\,\Vol_v(\rho),
\]
be the growth of the deficit across the window. Since the jumps of $\Vol_v$ only increase it, some radius $r$ in the window has cut-size at most the average increase of $\Vol_v$:
\begin{equation}\label{eq:ov-window}
|\delta_G(B_d(v,r))|\le\frac{\Vol_v(\rho)-\Vol_v(\rho-w)}{w}=\frac{\Vol_v(\rho)}{w}\bigl(1-e^{-\Delta}\bigr)\le\frac{\Vol_v(\rho)}{w}\,\min\{1,\Delta\}.
\end{equation}
Thus cutting at the best radius in the window costs at most $\Vol_v(\rho)/w$ per unit of deficit growth across the window, and no logarithm appears. Two requirements turn this into the charging scheme above.
\begin{enumerate}[label=(\alph*),leftmargin=*]
\item \emph{The window must be wide enough.} The cost per unit of growth must be $O(\frac{\LP(G)}{n})$, so the window must have width $w=\Omega\big(\frac{n}{\LP(G)}\Vol_v(\rho)\big)$. Heavy balls need wide windows, and light balls only thin ones.
\item \emph{The deficits of later pieces must increase by $\Delta$.} Every later piece should have deficit at least $F_v(\rho)+\Delta=F_v(\rho-w)$, that is, volume at most $\Vol_v(\rho-w)$. Otherwise the growth $\Delta$ does not show up in the budget.
\end{enumerate}

\paragraph{Global Carving.}
These two requirements lead directly to the algorithm. All remaining centers share a single radius $\rho$, which starts at $\frac n{10}$ and only decreases; as before, volumes are measured in the graph induced by the remaining vertices. \emph{Global Carving} repeats the following step while vertices remain.
\begin{enumerate}[label=(\arabic*),leftmargin=*]
\item Let $K$ be the largest volume of a ball of radius $\rho$, and set $w=\frac{n}{160}\cdot\frac{K}{\LP(G)}$.
\item Let $v$ be a vertex whose ball of radius $\rho-w$ has the largest volume.
\item Cut out the ball around $v$ at the radius in the window $(\rho-w,\rho)$ with the smallest cut-size, remove it, and set $\rho\leftarrow\rho-w$.
\end{enumerate}
Step~(1) chooses the width of the window as requirement~(a) asks: $w$ is proportional to the volume $K$ of the heaviest ball, and up to a constant it is the radius $\frac n{10}$ times the fraction of the total volume that sits in this ball. Steps~(2) and~(3) are designed for requirement~(b): $v$ has the heaviest ball of the inner radius $\rho-w$, and every later cut is made at a radius below $\rho-w$. The radius in Step~(3) is found by plain averaging, as in \eqref{eq:ov-window}; there is no Seymour step and no stability test.

We now show that the pieces, in the order in which they are carved, form an $O(1)$-weak local partition. We first prove the prefix bound, in three steps, and then check that the radius does not run out.

\paragraph{The level.}
The charging scheme above needs nondecreasing deficits. Instead of the deficits of the pieces, we use a quantity that is defined at the start of every step. Number the steps $j=1,2,\ldots$, write $K_j$, $w_j$ and $v_j$ for the quantities $K$, $w$ and $v$ of step $j$, write $\Delta_j$ for the growth of the deficit of $v_j$ across its window, as in \eqref{eq:ov-window}, and let $X_j$ be the piece carved in step $j$. The \emph{level} at the start of step $j$ is
\[
L_j=\log\frac{2e\,\LP(G)}{K_j},
\]
the deficit of the heaviest ball of the current radius. Since all volumes are at most $2\LP(G)$, we have $L_1\ge1$. The level has the two properties that the charging scheme needs.
\begin{itemize}[leftmargin=*]
\item \emph{It never decreases.} Between two steps, vertices are removed and the radius decreases. Both can only shrink balls, so $K_{j+1}\le K_j$, that is, $L_{j+1}\ge L_j$.
\item \emph{It is at most the deficit of the current piece.} The piece $X_j$ is a ball of radius less than the current radius, so its volume is at most $K_j$, and its deficit satisfies $F_j\ge L_j$.
\end{itemize}

\paragraph{The cost of a step.}
Fix a step $j$, and let $\rho$ be the radius at its start. Since $\Vol_{v_j}(\rho)\le K_j$, the window bound \eqref{eq:ov-window} shows that step $j$ cuts at most
\[
\frac{K_j}{w_j}\min\{1,\Delta_j\}=160\,\frac{\LP(G)}{n}\min\{1,\Delta_j\}
\]
edges of the remaining graph. This is where requirement~(a) is used: since $w_j$ is proportional to $K_j$, the ratio $K_j/w_j$ is the same in every step. We now show that the level rises by at least $\Delta_j$ in this step; this is where requirement~(b) is used. Step $j+1$ starts at radius $\rho-w_j$, and before $X_j$ was removed, $v_j$ had the heaviest ball of this radius. Removing $X_j$ only shrinks balls, so
\[
K_{j+1}\le\Vol_{v_j}(\rho-w_j)=e^{-\Delta_j}\,\Vol_{v_j}(\rho)\le e^{-\Delta_j}K_j,
\]
that is, $L_{j+1}\ge L_j+\Delta_j$ (see \cref{fig:paced-step}). Combining the two bounds, step $j$ cuts at most $160\frac{\LP(G)}{n}\min\{1,L_{j+1}-L_j\}$ edges of the remaining graph.

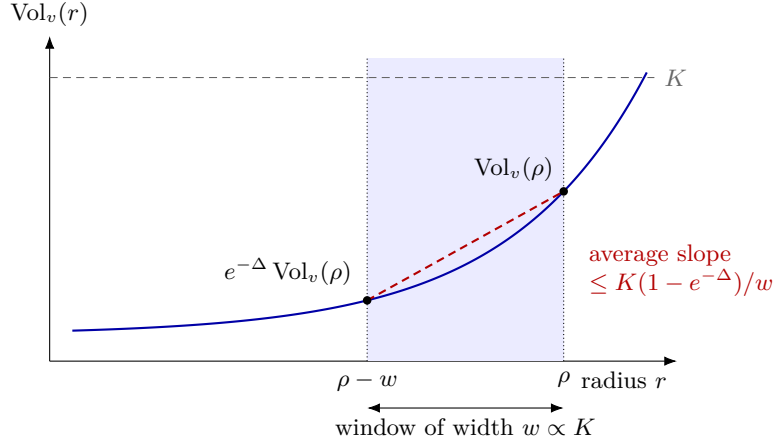
\begin{figure}[t]
\centering
\begin{tikzpicture}[x=1cm,y=1cm,>=Latex]
  \draw[->] (0,0) -- (8.3,0) node[below left,font=\footnotesize] {radius $r$};
  \draw[->] (0,0) -- (0,4.3) node[above,font=\footnotesize] {$\Vol_v(r)$};
  \fill[blue!8] (4.2,0) rectangle (6.8,4.0);
  \draw[thick,blue!65!black,domain=0.3:7.9,samples=80,smooth] plot (\x,{0.35+0.045*exp(0.55*\x)});
  \coordinate (A) at (4.2,{0.35+0.045*exp(0.55*4.2)});
  \coordinate (B) at (6.8,{0.35+0.045*exp(0.55*6.8)});
  \draw[densely dashed,red!70!black,thick] (A) -- (B);
  \fill (A) circle (1.6pt); \fill (B) circle (1.6pt);
  \draw[densely dotted] (4.2,0) -- (4.2,4.0);
  \draw[densely dotted] (6.8,0) -- (6.8,4.0);
  \node[below,font=\footnotesize] at (4.2,0) {$\rho-w$};
  \node[below,font=\footnotesize] at (6.8,0) {$\rho$};
  \draw[densely dashed,black!60] (0,3.75) -- (8.0,3.75) node[right,font=\footnotesize] {$K$};
  \draw[<->] (4.2,-0.62) -- (6.8,-0.62) node[midway,below,font=\footnotesize] {window of width $w\propto K$};
  \node[above left=2pt,font=\footnotesize] at (A) {$e^{-\Delta}\Vol_v(\rho)$};
  \node[above left,font=\footnotesize] at (B) {$\Vol_v(\rho)$};
  \node[font=\footnotesize,red!70!black,align=left,anchor=west] at (7.0,1.2) {average slope\\$\le K(1-e^{-\Delta})/w$};
\end{tikzpicture}
\caption{One step of Global Carving. Some radius in the window has cut-size at most the average slope of $\Vol_v$ across it. The new heaviest volume is at most $\Vol_v(\rho-w)=e^{-\Delta}\Vol_v(\rho)\le e^{-\Delta}K$, so the level rises by at least $\Delta$.}
\label{fig:paced-step}
\end{figure}
\paragraph{The prefix bound.}
Now fix $i$, and consider the prefix $X_1,\ldots,X_i$. As in \cref{sec:ov-difficulty}, every edge between two pieces is cut, in the remaining graph, when the earlier of the two pieces is removed, and it is counted at most twice in $\sum_{j\le i}|\delta_G(X_j)|$. The steps $j<i$ cut at most $160\frac{\LP(G)}{n}(L_{j+1}-L_j)$ edges each, and these bounds add up to $160\frac{\LP(G)}{n}(L_i-L_1)$. Step $i$ itself cuts at most $160\frac{\LP(G)}{n}\le160\frac{\LP(G)}{n}L_1$ edges. Hence
\[
\sum_{j\le i}|\delta_G(X_j)|\le2\cdot160\,\frac{\LP(G)}{n}\bigl((L_i-L_1)+L_1\bigr)=320\,\frac{\LP(G)}{n}\,L_i\le320\,\frac{\LP(G)}{n}\,F_i,
\]
where the last step uses that the level is at most the deficit of the current piece. This is the weak-local-partition bound with $\alpha=320$.

\paragraph{The radius does not run out.}
Every cut is made at a radius below $\frac n{10}$, so by \cref{sec:ov-region} every piece has at most $\frac n2$ vertices. It remains to check that the windows fit into the range of radii. We show that their total width is at most $\frac n{20}$, so that the radius never drops below $\frac n{20}$. Let $W_j$ be the volume of $X_j$ when it is carved. If $\rho$ is the radius at the start of step $j$, the piece $X_j$ contains the ball of radius $\rho-w_j$ around $v_j$, so $W_j\ge\Vol_{v_j}(\rho-w_j)$. If this volume is at least $K_j/2$, then $K_j\le2W_j$. Otherwise, since $K_{j+1}\le\Vol_{v_j}(\rho-w_j)$ by the cost paragraph, the heaviest volume drops below $K_j/2$, and then $K_j\le2(K_j-K_{j+1})$. In both cases,
\[
K_j\le2W_j+2(K_j-K_{j+1}).
\]
The volumes $W_j$ add up to at most $2\LP(G)$, as noted in \cref{sec:ov-difficulty}, and the differences $K_j-K_{j+1}$ add up to at most $K_1\le2\LP(G)$. So $\sum_jK_j\le8\LP(G)$, and $\sum_jw_j=\frac{n}{160\,\LP(G)}\sum_jK_j\le\frac n{20}$. In words, a wide window is paid for either by a heavy piece or by a large drop of the heaviest volume, and both are limited by the total volume. So the wide windows of requirement~(a) still fit into the range of radii.

\paragraph{Summary.}
The carving order of Global Carving is an $O(1)$-weak local partition, and together with the recursion of \cref{sec:ov-framework}, this gives an $O(\log n)$-approximation for \textsc{Cutwidth}. Both algorithms are instances of Greedy Region Growing: every step carves around the center with the heaviest ball (in Greedy Seymour Carving, each ball is measured at the outer end of its own annulus), so this largest volume never increases, and this is what controls the prefixes. They differ only in how the radius of each cut is chosen. Greedy Seymour Carving chooses it locally, for each center, so that every piece pays for its own boundary; this costs a factor $O(\log\log n)$. Global Carving uses a single radius shared by all centers, so that every cut is paid for by the rise of the level; the deficit $F_i$ in the final bound is at least the level, which accumulates the rises of all earlier steps, and the volume of a piece is used only to pay for the radius.

\subsection{Beyond cutwidth}\label{sec:ov-other}
The same greedy principle drives all our results. Each of the other problems brings one new difficulty.

\paragraph{Pathwidth and treedepth.} Here we must cut vertices rather than edges. Using the vertex relaxation of~\cite{BKS24}, Global Carving goes through with vertex separators in place of edge boundaries, and this extension is essentially direct.

\paragraph{Directed layouts.} Now the metric is asymmetric, and the ordering must respect precedence constraints. This is exactly where the approach of~\cite{BKS24} breaks down, as it relies on packing and covering properties of symmetric metrics. We need to carve pieces that are consistent with the precedence constraints and to arrange them in a topological order without losing the prefix bound. We do this by carving two kinds of balls: \emph{source balls}, which only have outgoing arcs at the time they are removed, and \emph{sink balls}, which only have incoming arcs. Placing the source balls at the beginning of the ordering and the sink balls at the end, in reverse order of removal, gives a topological order. Global Carving runs with these balls.

\paragraph{Flow-expander decomposition.} Here we need a single cut whose cost is paid for by the \emph{weight} of the smaller side, rather than by its LP volume. With a node-weighting $A$ (such as the degrees), we need a Seymour-type bound
\[
|\delta(S)|\le O(\phi\log\log|A|)\,A(S)\log\frac{|A|}{A(S)}
\]
for some set $S$ with $0<A(S)\le|A|/2$, where $|A|=A(V)$; the metric now comes from the dual of a multicommodity flow LP. Region growing only pays for cuts with LP volume, so the carving has to control weight and volume at the same time. We do this by requiring each annulus to be stable for both quantities, using the weight to order the pieces and the volume to pay for their boundaries, and then an averaging argument finds one piece satisfying the bound. This averaging argument needs a bound on the sum of the costs of all pieces, each relative to its own weight, rather than a bound on every prefix. This is why we use Greedy Seymour Carving here: Global Carving controls only prefixes, and we do not know how to remove the factor $O(\log\log n)$ in this setting.

\subsection{Organization}\label{sec:ov-organization}
\Cref{sec:prelim} collects preliminaries, including the region-growing, stable-annulus and persistent-pointer lemmas of \cref{sec:prelim-rg}, which are used in \cref{sec:expander}. \Cref{sec:cutwidth} gives the algorithm for \textsc{Cutwidth}, Global Carving, in full detail. \Cref{sec:expander} treats flow-expander decomposition with Greedy Seymour Carving, \cref{sec:directed} \textsc{Minimum Storage}, and \cref{sec:vertex} \textsc{Pathwidth} and \textsc{Treedepth}; the last two follow the structure of \cref{sec:cutwidth}. We conclude with open problems in \cref{sec:conclusion}. In the appendices, \cref{app:bks-wlp} compares our weak local partitions with those of~\cite{BKS24}, \cref{app:register} reduces \textsc{Minimum Register Sufficiency} to weighted \textsc{Minimum Storage}, \cref{app:nonnegative} extends the flow-expander decomposition to arbitrary nonnegative node-weightings, \cref{app:implementation} contains the implementation details: separation oracles, rational arithmetic, and the multicommodity flow LP, and \cref{app:gaps} shows the integrality gaps. \Cref{app:ai} is the AI disclosure.

For completeness and ease of reading, we keep \cref{sec:cutwidth,sec:expander,sec:directed,sec:vertex} largely self-contained. As a result, some lemmas appear in several sections in closely analogous forms: they rest on the same principle, but their details differ from problem to problem.

\section{Preliminaries}\label{sec:prelim}
All graphs are finite. Unweighted undirected graphs are simple. Edge weights, capacities, and algorithmic parameters are rational. Directed graphs are acyclic whenever a topological ordering is required. Unless stated otherwise, $n$ and $m$ denote the numbers of vertices and edges of the input graph. Logarithms are natural unless a base is specified.

\subsection{Notation}
For an undirected graph $G=(V,E)$ and $S\subseteq V$, let $G[S]$ be the induced subgraph and let
\[
\delta_G(S)=\set{uv\in E:|\{u,v\}\cap S|=1}.
\]
For a directed graph $G=(V,E)$, $\delta_G(S)$ denotes all arcs with exactly one endpoint in $S$. We use $\cc(H)$ for the vertex sets of the connected components of $H$, and $[n]=\{1,\ldots,n\}$.

A metric may assign distance zero to distinct vertices. A quasi-metric satisfies the directed triangle inequality but need not be symmetric. A symmetric demand $D$ assigns a nonnegative value $D(u,v)$ to each unordered pair $\{u,v\}$. A routing sends this amount between $u$ and $v$. The load of an edge is the sum of the commodity flows using it, and the congestion is the maximum ratio of edge load to edge capacity.

\subsection{Vertex orderings and path decompositions}
For an ordering $\pi=(v_1,\ldots,v_n)$, let $S_i^\pi=\{v_1,\ldots,v_i\}$. Its cutwidth is $\max_i|\delta_G(S_i^\pi)|$. Its vertex-separation cost is
\[
\VS_\pi(G)=\max_i\abs{\set{v_j:j\le i,\ v_jv_k\in E(G)\text{ for some }k>i}}.
\]
Restricting an ordering to an induced subgraph cannot increase either cost. The same observation holds for \textsc{Minimum Storage} (\cref{sec:directed}) when the ordering is topological.

A path decomposition is a sequence of bags $(X_1,\ldots,X_m)$ whose union is $V(G)$, such that every edge has both endpoints in a common bag and the bags containing any fixed vertex form an interval. Its width is $\max_i|X_i|-1$. The pathwidth $\pw(G)$ is the minimum width of a path decomposition.

\begin{lemma}\label{lem:vs-pw}
The minimum vertex-separation cost of an ordering of $G$ equals $\pw(G)$~\cite{Kin92}. An ordering of vertex-separation cost $k$ can be converted in polynomial time into a path decomposition of width at most $k$.
\end{lemma}
\begin{proof}
Let $\pi=(v_1,\ldots,v_n)$ have vertex-separation cost $k$. For each $i\in[n]$, set
\[
X_i=\{v_i\}\cup\set{v_j:j<i,\ v_jv_\ell\in E(G)\text{ for some }\ell\ge i}.
\]
The second set is the active set after the first $i-1$ vertices, so $|X_i|\le k+1$. If $v_jv_\ell$ is an edge with $j<\ell$, then both endpoints belong to $X_\ell$. Let $r_j$ be the largest index of a neighbor of $v_j$ appearing after $v_j$, or $j$ if no such neighbor exists. The bags containing $v_j$ are exactly $X_j,\ldots,X_{r_j}$. Thus $(X_1,\ldots,X_n)$ is a path decomposition of width at most $k$.

Conversely, let $(X_1,\ldots,X_m)$ have width $k$, and let $\ell(v)$ be the first bag containing $v$. Order the vertices by nondecreasing $\ell(v)$, breaking ties arbitrarily. Fix a proper prefix, and let $t$ be the first-bag index of its next vertex. If a processed vertex $u$ has an unprocessed neighbor $v$, then the intervals for $u$ and $v$ intersect. Since $\ell(u)\le t\le\ell(v)$, the interval for $u$ contains $t$. Hence every active processed vertex lies in $X_t$. The next unprocessed vertex also lies in $X_t$, so at most $|X_t|-1\le k$ processed vertices are active.
\end{proof}

The LPs in this paper use spreading constraints. They are satisfied by the distances of every ordering.

\begin{lemma}[Line metrics are spreading]\label{lem:line-spreading}
Let $\pi$ be an ordering of a finite set $V$, and let $d_\pi(u,v)=|\pi(u)-\pi(v)|$. Then $\sum_{v\in S}d_\pi(u,v)\ge(|S|^2-1)/4$ for every $u\in V$ and every $S\subseteq V$.
\end{lemma}
\begin{proof}
Let $s=|S|$; we may assume $u\in S$, since replacing a vertex of $S$ by $u$ does not increase the sum. Let $q_1\le\cdots\le q_{s-1}$ be the distances from $u$ to the vertices of $S\setminus\{u\}$. At most two vertices lie at each positive integer distance from $u$, so $q_k\ge\lceil k/2\rceil$. If $s-1=2k$, then $\sum_{j=1}^{s-1}\lceil j/2\rceil=k(k+1)=(s^2-1)/4$; if $s-1=2k+1$, it equals $(k+1)^2=s^2/4$.
\end{proof}

\subsection{Region growing and stable annuli}\label{sec:prelim-rg}
For a metric $d$, a center $u$, and $r>0$, write
\[
B_d(u,r)=\{v:d(u,v)<r\}.
\]
All our algorithms grow balls whose volume behaves like the volume in \cref{sec:ov-region}. We isolate the properties we use. A \emph{volume function} is a function $\mathcal{V}:(0,\infty)\to(0,\infty)$ that is nondecreasing and piecewise affine with finitely many breakpoints, called \emph{events}, at which it may jump upward, and whose derivative away from the events is a nonnegative quantity $c(r)$, called the \emph{boundary cost} at $r$. Given a normalization $\Lambda\ge\sup_r\mathcal{V}(r)$, the \emph{deficit} of $\mathcal{V}$ is
\[
F(r)=\log\frac{e\Lambda}{\mathcal{V}(r)}.
\]
It is nonincreasing in $r$ and at least one. An interval $[a,b]$ is \emph{stable} if $F(a)\le 2F(b)$.

\begin{lemma}[Region growing]\label{lem:rg}
Let $\mathcal{V}$ be a volume function and $0<a<b$. There is an event-free radius $r\in(a,b)$ with
\[
c(r)\le\frac{k}{b-a}\,\mathcal{V}(r),\qquad k=\Big\lceil\log_2\frac{2\mathcal{V}(b)}{\mathcal{V}(a)}\Big\rceil.
\]
\end{lemma}
\begin{proof}
This is the argument of \cref{sec:ov-region}. Divide $[a,b]$ into $k$ intervals of equal length. If $\mathcal{V}$ more than doubled across each of them, then $\mathcal{V}(b)>2^k\mathcal{V}(a)\ge2\mathcal{V}(b)$. So some interval $[a',b']$ has $\mathcal{V}(b')\le2\mathcal{V}(a')$. Since the jumps are nonnegative, $\int_{a'}^{b'}c(r)\,dr\le \mathcal{V}(b')-\mathcal{V}(a')\le \mathcal{V}(a')$, so some event-free $r\in(a',b')$ has $c(r)\le \mathcal{V}(a')/(b'-a')\le k\mathcal{V}(r)/(b-a)$.
\end{proof}

\begin{lemma}[Stable annulus]\label{lem:stable-annulus}
Let $\mathcal{V}$ be a volume function with deficit $F$, and let $[a,b]$ be stable. There is an event-free radius $r\in(a,b)$ with
\[
c(r)\le\frac{4}{b-a}\,\mathcal{V}(r)F(r).
\]
Such a radius, up to an absolute constant factor in the bound, can be found in polynomial time (\cref{app:implementation}).
\end{lemma}
\begin{proof}
Apply \cref{lem:rg}. By stability,
$k\le 1+\log_2\frac{2\mathcal{V}(b)}{\mathcal{V}(a)}=2+\frac{F(a)-F(b)}{\log 2}\le 2+\frac{F(b)}{\log2}$.
Since $1\le F(b)\le F(r)$, this is at most $(2+1/\log 2)F(r)\le 4F(r)$.
\end{proof}

\begin{corollary}[One-step carving]\label{lem:one-step}
Let $\mathcal{V}$ be a volume function whose deficit takes values in $[1,D]$, and let $0<a<b$. There is an event-free radius $r\in(a,b)$ with
\[
c(r)\le O\Big(\frac{1+\log D}{b-a}\Big)\mathcal{V}(r)F(r).
\]
\end{corollary}
\begin{proof}
Divide $[a,b]$ into $\kappa=\lceil\log_2D\rceil+1$ intervals of equal length. If none were stable, then $F(a)>2^\kappa F(b)\ge2^\kappa>D$. Apply \cref{lem:stable-annulus} to a stable one, whose length is $(b-a)/\kappa$.
\end{proof}

The algorithm of \cref{sec:expander} maintains, for each center, a pointer into a fixed sequence of annuli that only moves inward. The following lemma bounds how often a pointer moves.

\begin{lemma}[Persistent pointers]\label{lem:persistence}
Let $a_0>a_1>\cdots>a_\kappa$ be fixed radii, and let $F_1,F_2,\ldots$ be functions such that every $F_t$ is nonincreasing, $F_t\le F_{t+1}$ pointwise, and all values lie in $[1,D]$. A pointer $p\in\{1,\ldots,\kappa\}$ starts at $1$ and may only increase. Call a move of the pointer from $p$ to $p+1$ at time $t$ \emph{forced by $F$} if $F_t(a_p)>2F_t(a_{p-1})$. Then at most $\log_2D$ moves are forced by $F$.
\end{lemma}
\begin{proof}
We show that at the $k$-th forced move, $F_t(a_p)>2^k$. For $k=1$ this holds since $F_t(a_{p-1})\ge1$. Let the $(k-1)$-th forced move happen at time $t'$ from $p'$, so that $F_{t'}(a_{p'})>2^{k-1}$. The $k$-th forced move happens at a time $t\ge t'$ from some $p\ge p'+1$. Since $a_{p-1}\le a_{p'}$, monotonicity in the radius and in time give $F_t(a_{p-1})\ge F_{t'}(a_{p'})>2^{k-1}$, hence $F_t(a_p)>2^k$. All values are at most $D$, so $2^k<D$.
\end{proof}
In particular, if a pointer moves only when forced and $\kappa\ge\log_2D+1$, then it never passes the last annulus.

The stable-annulus and persistent-pointer lemmas are used only in \cref{sec:expander}, which uses Greedy Seymour Carving.

\section{An \texorpdfstring{$O(\log n)$}{O(log n)}-Approximation for Cutwidth}\label{sec:cutwidth}
For an undirected graph $G$, we use the spreading-metric relaxation from \cref{sec:ov-framework}~\cite{ENRS00,BKS24}:
\begin{equation}\label{eq:cw-lp}
\begin{aligned}
\LP(G)=\min\quad &\sum_{uv\in E(G)}d(u,v)\\
\text{s.t.}\quad &d(u,w)\le d(u,v)+d(v,w) &&u,v,w\in V(G),\\
&\sum_{v\in S}d(u,v)\ge\frac{|S|^2-1}{4} &&u\in V(G),\ S\subseteq V(G),\\
&d(u,v)=d(v,u) &&u,v\in V(G),\\
&d(u,u)=0,\quad d(u,v)\ge0 &&u,v\in V(G).
\end{aligned}
\end{equation}
Its spreading constraints admit a polynomial-time separation oracle (\cref{app:implementation}). For every induced subgraph $H$ of $G$,
\begin{equation}\label{eq:cw-density}
\LP(H)\le |V(H)|\cw(H)\le|V(H)|\cw(G).
\end{equation}
Indeed, let $\pi$ be an ordering of $H$ of cutwidth $\cw(H)$. The metric $d_\pi$ is feasible for \eqref{eq:cw-lp} on $H$ by \cref{lem:line-spreading}. An edge $uv$ has length $|\pi(u)-\pi(v)|$, which is the number of prefixes $S_i^\pi$ that it crosses, so the objective is $\sum_{i<|V(H)|}|\delta_H(S_i^\pi)|\le|V(H)|\cw(H)$. The second inequality holds since restricting an ordering of $G$ to $V(H)$ does not increase cutwidth.
Throughout this section, $G$ is a connected graph on $n\ge2$ vertices and $d$ is an optimal solution of \eqref{eq:cw-lp}. As in \cref{sec:ov-framework}, the \emph{volume} of a set $S\subseteq V(G)$ is
\begin{equation}\label{eq:static-volume}
\Vol(S)=\frac{\LP(G)}{n}|S|+\sum_{uv\in E(G[S])}d(u,v).
\end{equation}
The volumes of the parts of any partition sum to at most $2\LP(G)$.

\subsection{Weak local partitions}
\begin{definition}[Weak local partition]\label{def:wlp}
An ordered partition $(X_1,\ldots,X_k)$ of $V(G)$ is an \emph{$\alpha$-weak local partition} if $|X_i|\le n/2$ for every $i$ and
\begin{equation}\label{eq:wlp}
\sum_{j\le i}|\delta_G(X_j)|
\le\alpha\,\frac{\LP(G)}{n}\,\log\frac{2e\,\LP(G)}{\Vol(X_i)}
\qquad\text{for every }i\in[k].
\end{equation}
\end{definition}
This implies the unordered condition of~\cite{BKS24}; see \cref{app:bks-wlp}. The following lemma makes the recursion of \cref{sec:ov-framework} precise.

\begin{lemma}\label{lem:wlp-to-cw}
Suppose an $\alpha$-weak local partition can be computed in polynomial time for every connected induced subgraph with at least two vertices. Then \textsc{Cutwidth} admits a polynomial-time $O(\alpha\log n)$-approximation.
\end{lemma}
\begin{proof}
Since there are no edges between connected components, we may order the components one after another, so we may assume that $G$ is connected; a single vertex is ordered trivially. On a connected instance $H$ with at least two vertices, solve \eqref{eq:cw-lp} on $H$, compute an $\alpha$-weak local partition of $H$, recursively order the connected components of each of its parts, and concatenate the orders following the order of the parts. Every recursive instance has at most half as many vertices as its parent, so the recursion depth is $O(\log n)$. Write $\Vol_H$ for the volume \eqref{eq:static-volume} computed in $H$.

Fix a prefix $P$ of the final ordering. Every instance occupies a contiguous block of the ordering. Let $H_0=G,H_1,\ldots,H_T$ be the instances with at least two vertices whose block contains the last vertex of $P$. For each $t\le T$, let $Y_t$ be the part of $H_t$ containing that vertex; for $t<T$, the next instance $H_{t+1}$ is a component of $H_t[Y_t]$. An edge crossing $P$ that is separated at the level of $H_t$ leaves a part of $H_t$ that precedes or equals $Y_t$. By \eqref{eq:wlp} and \eqref{eq:cw-density}, there are at most
\[
\alpha\,\frac{\LP(H_t)}{|V(H_t)|}\log\frac{2e\,\LP(H_t)}{\Vol_{H_t}(Y_t)}
\le\alpha\cw(G)\log\frac{2e\,\LP(H_t)}{\Vol_{H_t}(Y_t)}
\]
such edges. Restricting the metric of $H_t$ to $H_{t+1}$ is feasible, so $\LP(H_{t+1})\le\Vol_{H_t}(Y_t)$ for $t<T$. At the last level, $\Vol_{H_T}(Y_T)\ge\LP(H_T)/|V(H_T)|$, since every vertex contributes $\LP(H_T)/|V(H_T)|$ to the volume. Therefore
\[
\sum_{t=0}^T\log\frac{2e\,\LP(H_t)}{\Vol_{H_t}(Y_t)}
\le T\log(2e)+\log\frac{\LP(H_0)}{\LP(H_T)}+\log\bigl(2e|V(H_T)|\bigr)=O(\log n).
\]
Here we used $T=O(\log n)$, $\LP(G)\le n\cw(G)\le n^3$, and $\LP(H_T)\ge3/8$. The last bound holds for every connected instance $H$ with $s\ge2$ vertices: spreading gives $\sum_vd(u,v)\ge(s^2-1)/4$, while the triangle inequality along a path gives $d(u,v)\le\LP(H)$, so $\LP(H)\ge(s^2-1)/(4s)\ge3/8$. Finally, every edge crossing $P$ is separated at one of the levels $H_0,\ldots,H_T$, since a single vertex contains no edge. This proves the claim.
\end{proof}

\subsection{The carving algorithm}\label{sec:cw-algorithm}
The algorithm grows balls in the graph induced by the remaining vertices. For $U\subseteq V(G)$, $u\in U$ and $r\in\R$, let $B^U(u,r)=\{v\in U:d(u,v)<r\}$, and define its volume as in \cref{sec:ov-region}, computed in $G[U]$:
\begin{equation}\label{eq:residual-volume}
\Vol_u^U(r)=\Vol\bigl(B^U(u,r)\bigr)
+\sum_{\substack{xy\in E(G[U])\\x\in B^U(u,r),\ y\notin B^U(u,r)}}
\bigl(r-d(u,x)\bigr).
\end{equation}
For $r\le0$ the ball is empty and its volume is $0$.

\begin{lemma}\label{lem:cw-volume}
For every $U\subseteq V(G)$, every $u\in U$, and every $r>0$:
\begin{enumerate}[label=(\roman*),leftmargin=*]
\item $\Vol_u^U$ is a volume function (\cref{sec:prelim-rg}) with boundary cost $|\delta_{G[U]}(B^U(u,r))|$ and normalization $2\LP(G)$; its events are the distances $d(u,x)$;
\item $\LP(G)/n\le\Vol_u^U(r)\le2\LP(G)$;
\item removing vertices other than $u$ from $U$ does not increase $\Vol_u^U(r)$;
\item if $r\le n/10$, then $|B^U(u,r)|\le n/2$.
\end{enumerate}
\end{lemma}
\begin{proof}
Write $B=B^U(u,r)$.
\begin{enumerate}[label=(\roman*),leftmargin=*]
\item Between two consecutive values of $d(u,\cdot)$, the ball $B$ does not change, and every edge $xy$ of $G[U]$ with $x\in B$ and $y\notin B$ contributes $r-d(u,x)$, which grows with slope one. So the derivative of $\Vol_u^U$ is $|\delta_{G[U]}(B)|$. When $r$ reaches $d(u,y)$ and $y$ enters the ball, the term $\LP(G)/n$ of $y$ is added, and for every edge $xy$ with $x$ already in the ball, the partial term $d(u,y)-d(u,x)$ is replaced by the full length $d(x,y)$; this is an increase, since $d(u,y)-d(u,x)\le d(x,y)$ by the triangle inequality. So $\Vol_u^U$ only jumps upward.
\item The vertex terms sum to at most $\LP(G)$, and the edge terms to at most $\sum_{xy\in E(G)}d(x,y)=\LP(G)$. Since $u\in B$, the volume is at least $\LP(G)/n$.
\item Removing a vertex $y\ne u$ from $U$ removes its own term and the terms of the edges incident to $y$, and leaves all other terms unchanged.
\item Suppose that the ball contains a set $X$ of size $s=\lfloor n/2\rfloor+1$. Then $s\ge2$ and $n\le2s-1$, so $\sum_{x\in X}d(u,x)<sn/10\le s(2s-1)/10\le(s^2-1)/4$, contradicting spreading; the last inequality is equivalent to $s^2+2s-5\ge0$.
\end{enumerate}
\end{proof}

\paragraph{The algorithm.}
\Cref{alg:cutwidth} is the Global Carving of \cref{sec:ov-paced}, with balls and volumes taken among the remaining vertices.

\begin{framed}
\noindent\refstepcounter{algorithm}\label{alg:cutwidth}\textbf{Algorithm \thealgorithm} (Global Carving)\\[2pt]
\textbf{Input:} a connected graph $G$ with $n\ge2$ vertices, and an optimal solution $d$ of \eqref{eq:cw-lp}.\\
\textbf{Output:} an ordered partition $(X_1,\ldots,X_k)$ of $V(G)$.
\begin{enumerate}[label=\arabic*:,leftmargin=2.2em,itemsep=1pt,topsep=4pt]
\item Initialize $U\leftarrow V(G)$, $\rho\leftarrow n/10$, and $i\leftarrow0$.
\item \textbf{while} $U\ne\emptyset$ \textbf{do}
\item \quad $i\leftarrow i+1$, $U_i\leftarrow U$, $\rho_i\leftarrow\rho$, $K_i\leftarrow\max_{u\in U}\Vol_u^U(\rho)$, and $w_i\leftarrow\frac{n}{160}\cdot\frac{K_i}{\LP(G)}$.
\item \quad Let $u_i\in U$ be a vertex maximizing $\Vol_u^U(\rho-w_i)$.
\item \quad Let $r_i\in(\rho-w_i,\rho)$ be a radius that is not an event of $\Vol_{u_i}^U$ and minimizes $|\delta_{G[U]}(B^U(u_i,r_i))|$.
\item \quad $X_i\leftarrow B^U(u_i,r_i)$, $U\leftarrow U\setminus X_i$, and $\rho\leftarrow\rho-w_i$.
\item \textbf{return} $(X_1,\ldots,X_k)$.
\end{enumerate}
\end{framed}

\noindent Line~5 is well defined as long as $\rho>0$, since then $K_i\ge\LP(G)/n>0$ by \cref{lem:cw-volume}(ii), and the boundary is constant between consecutive events; we show in \cref{lem:carving-props} that $\rho$ stays above $n/20$. The window of iteration $i$ is $(\rho_{i+1},\rho_i)$, where $\rho_{i+1}=\rho_i-w_i$, and we write $W_i=\Vol_{u_i}^{U_i}(r_i)$ for the volume of $X_i$ when it is carved.

\subsection{Analysis}
We follow \cref{sec:ov-paced}. The first lemma collects the facts about volumes that we use.

\begin{lemma}[Volumes]\label{lem:heaviest}
For every iteration $i$:
\begin{enumerate}[label=(\roman*),leftmargin=*]
\item $\Vol(X_i)\le W_i\le K_i$. In words, a piece is no heavier than its ball when carved, which is no heavier than the heaviest ball.
\item If $i<k$, then $K_{i+1}\le\Vol_{u_i}^{U_i}(\rho_{i+1})\le W_i$, and hence $W_{i+1}\le W_i$. In words, after iteration $i$, no ball of the new radius is heavier than the ball of $u_i$ at this radius.
\end{enumerate}
Moreover, $\sum_iW_i\le2\LP(G)$.
\end{lemma}
\begin{proof}
\begin{enumerate}[label=(\roman*),leftmargin=*]
\item Recall from \eqref{eq:static-volume} that $\Vol(X_i)$ counts $\LP(G)/n$ for every vertex of $X_i$ and the full length of every edge inside $X_i$. The volume $W_i$ of the ball $X_i=B^{U_i}(u_i,r_i)$ counts the same terms, plus part of the length of every edge leaving $X_i$ in $G[U_i]$; so $\Vol(X_i)\le W_i$. Since $\Vol_{u_i}^{U_i}$ is nondecreasing and $r_i<\rho_i$, the choice of $K_i$ in Line~3 gives $W_i\le\Vol_{u_i}^{U_i}(\rho_i)\le K_i$.
\item Since $\rho_{i+1}<r_i$, we have $\Vol_{u_i}^{U_i}(\rho_{i+1})\le W_i$. Every $u\in U_{i+1}$ satisfies $\Vol_u^{U_{i+1}}(\rho_{i+1})\le\Vol_u^{U_i}(\rho_{i+1})\le\Vol_{u_i}^{U_i}(\rho_{i+1})$, by \cref{lem:cw-volume}(iii) and the choice of $u_i$ in Line~4, and $K_{i+1}$ is the largest of these volumes. Finally, $W_{i+1}\le K_{i+1}$ by (i).
\end{enumerate}
For the sum, every vertex contributes $\LP(G)/n$ only to the volume of its own piece. An edge contributes at most its length, and only to the volume of the piece containing the endpoint removed first, since afterwards it is no longer an edge of the remaining graph. So $\sum_iW_i\le\LP(G)+\LP(G)$.
\end{proof}

\begin{lemma}[Radius]\label{lem:carving-props}
The radius never drops below $n/20$, that is, $\rho_{i+1}\ge n/20$ for every $i$. Consequently, \cref{alg:cutwidth} terminates after at most $n$ iterations, in polynomial time, and $|X_i|\le n/2$ for every $i$.
\end{lemma}
\begin{proof}
As in \cref{sec:ov-paced}, a wide window is paid for either by a heavy piece or by a large drop of the heaviest volume. Let $j<k$. If $\Vol_{u_j}^{U_j}(\rho_{j+1})\ge K_j/2$, then $K_j\le2W_j$ by \cref{lem:heaviest}(ii); otherwise, $K_{j+1}<K_j/2$ by the same lemma. In both cases,
\[
K_j\le2W_j+2(K_j-K_{j+1}).
\]
Fix $i$. Summing over $j<i$, adding $K_i$, and using $K_i\le K_1\le2\LP(G)$, we get
\[
\sum_{j\le i}K_j\le2\sum_{j<i}W_j+2K_1\le8\LP(G),
\qquad\text{so}\qquad
\rho_{i+1}=\frac n{10}-\frac{n}{160\,\LP(G)}\sum_{j\le i}K_j\ge\frac n{20}.
\]
Hence every window lies in $(n/20,n/10)$. Since $r_i<n/10$, we have $|X_i|\le n/2$ by \cref{lem:cw-volume}(iv). Since $r_i>0$, every iteration removes $u_i\in X_i$, so there are at most $n$ iterations, and each of them evaluates volumes at given radii and boundaries between consecutive events, which takes polynomial time.
\end{proof}

\begin{lemma}[Prefix bound]\label{lem:cw-prefix}
For every $i$,
\begin{equation}\label{eq:cw-prefix}
\sum_{j\le i}|\delta_G(X_j)|\le320\,\frac{\LP(G)}{n}\log\frac{2e\,\LP(G)}{W_i}.
\end{equation}
\end{lemma}
\begin{proof}
Every cut is paid for by the rise of the level, as in \cref{sec:ov-paced}. Let $L_j=\log\frac{2e\,\LP(G)}{K_j}$ be the level at the start of iteration $j$; it is at least one by \cref{lem:cw-volume}(ii).

\emph{The cost of a step.} Fix an iteration $j$, and let
\[
\Delta_j=\log\frac{\Vol_{u_j}^{U_j}(\rho_j)}{\Vol_{u_j}^{U_j}(\rho_{j+1})}\ge0
\]
be the growth of the deficit of $u_j$ across its window; it is finite since $\rho_{j+1}>0$ by \cref{lem:carving-props}. By \cref{lem:cw-volume}(i), $\Vol_{u_j}^{U_j}$ only jumps upward, and away from the events its derivative is the boundary of the ball in $G[U_j]$. So the integral of the boundary over the window is at most $\Vol_{u_j}^{U_j}(\rho_j)-\Vol_{u_j}^{U_j}(\rho_{j+1})$, and Line~5 chooses a radius with the smallest boundary. Since $\Vol_{u_j}^{U_j}(\rho_j)\le K_j$, iteration $j$ cuts
\[
|\delta_{G[U_j]}(X_j)|
\le\frac{\Vol_{u_j}^{U_j}(\rho_j)}{w_j}\bigl(1-e^{-\Delta_j}\bigr)
\le\frac{K_j}{w_j}\min\{1,\Delta_j\}
=160\,\frac{\LP(G)}{n}\min\{1,\Delta_j\}
\]
edges. If $j<k$, \cref{lem:heaviest} gives $K_{j+1}\le\Vol_{u_j}^{U_j}(\rho_{j+1})=e^{-\Delta_j}\Vol_{u_j}^{U_j}(\rho_j)\le e^{-\Delta_j}K_j$, so the level rises by at least $\Delta_j$: $L_{j+1}\ge L_j+\Delta_j$.

\emph{The prefix bound.} Now fix $i$. For $j<i$, iteration $j$ cuts at most $160\frac{\LP(G)}{n}(L_{j+1}-L_j)$ edges, and these bounds add up to $160\frac{\LP(G)}{n}(L_i-L_1)$. Iteration $i$ cuts at most $160\frac{\LP(G)}{n}\le160\frac{\LP(G)}{n}L_1$ edges. An edge between $X_p$ and $X_q$ with $p<q$ lies in $\delta_{G[U_p]}(X_p)$, and it is counted at most twice on the left-hand side of \eqref{eq:cw-prefix}, and only if $p\le i$. Hence
\[
\sum_{j\le i}|\delta_G(X_j)|
\le2\sum_{j\le i}|\delta_{G[U_j]}(X_j)|
\le320\,\frac{\LP(G)}{n}\,L_i
\le320\,\frac{\LP(G)}{n}\log\frac{2e\,\LP(G)}{W_i},
\]
where the last step uses $W_i\le K_i$.
\end{proof}

\begin{lemma}\label{lem:carving-wlp}
The carving order is a $320$-weak local partition.
\end{lemma}
\begin{proof}
The size condition is part of \cref{lem:carving-props}. Since $\Vol(X_i)\le W_i$ by \cref{lem:heaviest}, \cref{lem:cw-prefix} gives \eqref{eq:wlp} with $\alpha=320$.
\end{proof}

\begin{proof}[Proof of the cutwidth statement in \cref{thm:main-layout}]
Combine \cref{lem:wlp-to-cw,lem:carving-wlp}. \Cref{alg:cutwidth} runs in polynomial time by \cref{lem:carving-props}; the LP and the numerical implementation are discussed in \cref{app:implementation}.
\end{proof}

\section{Flow-Expander Decomposition}\label{sec:expander}
Let $G$ be an undirected graph, and let $A:V(G)\to\R_{\ge1}$ be a rational node-weighting; so every vertex has weight at least one. Write $A(Y)=\sum_{v\in Y}A(v)$. A demand $D$ on the vertex pairs is \emph{$A$-respecting} if $\sum_{v\ne u}D(u,v)\le A(u)$ for every vertex $u$. Following~\cite{BJS26}, we say that $A$ is \emph{$\phi$-flow-expanding in $G$} if every $A$-respecting demand can be routed in $G$ with congestion at most $1/\phi$. For $A=\deg_G$, this says that $G$ is a $\phi$-flow expander. To avoid clutter, in this section $\log\log x$ stands for $\max\{1,\log\log x\}$. We prove the following weighted version of \cref{thm:main-expander}.

\begin{theorem}\label{thm:weighted-expander}
Let $G$ be a graph with a node-weighting $A$ as above, and write $|A|=A(V(G))$. For every rational $\phi>0$, one can delete in polynomial time $O(\phi|A|\log|A|\log\log|A|)$ edges so that $A$ is $\phi$-flow-expanding in every remaining connected component.
\end{theorem}
The theorem can be extended to edge-weighted graphs, with the total weight of the deleted edges in place of their number. Moreover, for every rational node-weighting $A\ge0$, one can delete $O(\phi|A|\log n\log\log n)$ edges so that $A$ is $\phi$-flow-expanding in every remaining component (\cref{app:nonnegative}).

\begin{proof}[Proof of \cref{thm:main-expander}]
Isolated vertices are trivially $\phi$-flow expanders, so we may remove them and let $A=\deg_G$, so that $|A|=2m\le2n^2$ and $\log|A|=O(\log n)$. If $H$ is a remaining component, every $\deg_H$-respecting demand is also $A$-respecting, so $H$ is a $\phi$-flow expander.
\end{proof}

\subsection{The flow LP and its dual}
Let $G$ be a connected graph, and write $|A|=A(V(G))$. The \emph{product demand} is $\Pi(u,v)=A(u)A(v)/|A|$. It is $A$-respecting, since every vertex $v$ has total demand $\sum_{u\ne v}\Pi(u,v)\le A(v)$; summing over $v$, the total demand $\sum_{\{u,v\}}\Pi(u,v)$ over all pairs is at most $|A|/2$.

\begin{lemma}\label{lem:product-certificate}
If $\Pi$ can be routed in $G$ with congestion at most $1/(2\phi)$, then $A$ is $\phi$-flow-expanding in $G$, that is, every $A$-respecting demand can be routed in $G$ with congestion at most $1/\phi$~\cite[Fact~2.1]{BJS26}.
\end{lemma}
\begin{proof}
Let $D$ be an $A$-respecting demand. Route each pair $u,v$ through every intermediate vertex $z$, sending an $A(z)/|A|$ fraction of $D(u,v)$ along the two legs $u$--$z$ and $z$--$v$. The total amount that has to be sent between a pair $x,z$ in this way is at most
\[
\frac{A(z)}{|A|}\sum_yD(x,y)+\frac{A(x)}{|A|}\sum_yD(z,y)\le\frac{2A(x)A(z)}{|A|}=2\Pi(x,z).
\]
So two copies of a routing of $\Pi$ carry all legs, and $D$ is routed with congestion at most $1/\phi$.
\end{proof}
The minimum congestion for routing $\Pi$ is the value of a multicommodity flow LP of polynomial size (\cref{app:implementation}). If it exceeds $1/(2\phi)$, LP duality gives nonnegative edge lengths $\ell$ such that
\begin{equation}\label{eq:failed-dual}
\sum_{e\in E(G)}\ell_e\le1
\qquad\text{and}\qquad
\sum_{\{u,v\}}\Pi(u,v)\,d(u,v)>\frac1{2\phi},
\end{equation}
where $d$ is the shortest-path metric of $G$ under $\ell$. The main lemma rounds these lengths into a cut.

\begin{lemma}\label{lem:size-cut}
Given lengths satisfying \eqref{eq:failed-dual}, one can find in polynomial time a set $S\subseteq V(G)$ with $0<A(S)\le|A|/2$ and
\begin{equation}\label{eq:size-cut}
|\delta_G(S)|
=O(\phi\log\log|A|)\,A(S)\log\frac{|A|}{A(S)}.
\end{equation}
\end{lemma}

\subsection{Proof of \texorpdfstring{\cref{lem:size-cut}}{the main lemma}}
Fix lengths satisfying \eqref{eq:failed-dual}, and let $R=1/(16\phi|A|)$. The easy case is when some ball of radius $R$ carries a third of the total weight. Then a much stronger guarantee holds: a single threshold cut around the ball has at most $O(\phi)A(S)$ edges, without the logarithmic factor of \eqref{eq:size-cut}. This is essentially~\cite[Lemma~4.1]{BJS26}, and a similar argument appears in~\cite{LR99}; we include a short proof.

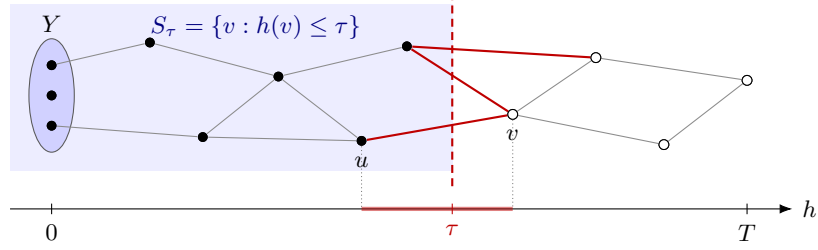
\begin{figure}[t]
\centering
\begin{tikzpicture}[x=1cm,y=1cm,>=Latex,
  vin/.style={circle,fill=black,inner sep=1.3pt},
  vout/.style={circle,draw=black,fill=white,inner sep=1.3pt},
  ed/.style={draw=black!45},
  cut/.style={draw=red!75!black,thick}]
  \fill[blue!7] (-0.55,-0.95) rectangle (5.3,1.25);
  \node[font=\footnotesize,blue!50!black,anchor=north] at (2.7,1.25) {$S_\tau=\{v:h(v)\le\tau\}$};
  \draw[black!65,fill=blue!18] (0,0.05) ellipse (0.3 and 0.75);
  \node[font=\footnotesize] at (0,1.0) {$Y$};
  \node[vin] (x1) at (0,0.45) {};
  \node[vin] (x2) at (0,0.05) {};
  \node[vin] (x3) at (0,-0.35) {};
  \node[vin] (a) at (1.3,0.75) {};
  \node[vin] (b) at (2.0,-0.5) {};
  \node[vin] (c) at (3.0,0.3) {};
  \node[vin,label={[font=\footnotesize]below:$u$}] (d) at (4.1,-0.55) {};
  \node[vin] (e) at (4.7,0.7) {};
  \node[vout,label={[font=\footnotesize]below:$v$}] (f) at (6.1,-0.2) {};
  \node[vout] (g) at (7.2,0.55) {};
  \node[vout] (k) at (8.1,-0.6) {};
  \node[vout] (m) at (9.2,0.25) {};
  \draw[ed] (x1)--(a) (x3)--(b) (a)--(c) (b)--(c) (c)--(d) (c)--(e) (d)--(b) (f)--(g) (f)--(k) (g)--(m) (k)--(m);
  \draw[cut] (d)--(f) (e)--(f) (e)--(g);
  \foreach \n in {x1,x2,x3,a,b,c,d,e} \node[vin] at (\n) {};
  \foreach \n in {f,g,k,m} \node[vout] at (\n) {};
  \draw[densely dashed,red!75!black,thick] (5.3,-1.15) -- (5.3,1.35);
  % axis
  \draw[->] (-0.55,-1.45) -- (9.8,-1.45) node[right,font=\footnotesize] {$h$};
  \draw (0,-1.38)--(0,-1.52) node[below,font=\footnotesize] {$0$};
  \draw (9.2,-1.38)--(9.2,-1.52) node[below,font=\footnotesize] {$T$};
  \draw[red!75!black] (5.3,-1.38)--(5.3,-1.52) node[below,font=\footnotesize] {$\tau$};
  \draw[densely dotted,black!50] (d) -- (4.1,-1.45);
  \draw[densely dotted,black!50] (f) -- (6.1,-1.45);
  \draw[line width=1.6pt,red!75!black,opacity=0.5] (4.1,-1.45) -- (6.1,-1.45);
\end{tikzpicture}
\caption{The threshold cut in the proof of \cref{lem:heavy-ball}. Every vertex $v$ is placed at its distance $h(v)=d(v,Y)$ from $Y$, which sits at $0$. For a threshold $\tau\in[0,T)$, the set $S_\tau$ (shaded) contains $Y$. An edge $uv$ is cut (red) only if $\tau$ falls between $h(u)$ and $h(v)$ (the highlighted interval); for uniform $\tau$, this happens with probability at most $|h(u)-h(v)|/T\le\ell_{uv}/T$.}
\label{fig:threshold}
\end{figure}

\begin{lemma}\label{lem:heavy-ball}
Let $Y\subseteq V(G)$ have $A(Y)\ge|A|/3$ and diameter at most $2R$. One can find in polynomial time a set $S\subseteq V(G)$ with $0<A(S)\le|A|/2$ and $|\delta_G(S)|\le8\phi\,A(S)$.
\end{lemma}
\begin{proof}
Let $h(v)=d(v,Y)$ be the distance of $v$ from $Y$. We first show that, weighted by $A$, the vertices are far from $Y$ on average. Any two vertices $u,v$ are connected through their nearest points in $Y$, and $Y$ has diameter at most $2R$, so $d(u,v)\le h(u)+h(v)+2R$. By \eqref{eq:failed-dual},
\[
\frac1{2\phi}<\sum_{\{u,v\}}\Pi(u,v)\,d(u,v)
\le\sum_{\{u,v\}}\Pi(u,v)\bigl(h(u)+h(v)\bigr)+2R\sum_{\{u,v\}}\Pi(u,v)
\le\sum_vA(v)h(v)+R|A|,
\]
where the last step uses that every vertex has total demand at most $A(v)$ and that the total demand is at most $|A|/2$. Since $R|A|=\frac1{16\phi}$, we get
\[
\sum_vA(v)h(v)>\frac1{4\phi}.
\]
In particular, $T=\max_vh(v)>0$. Pick $\tau\in[0,T)$ uniformly at random, and let $S_\tau=\{v:h(v)\le\tau\}$, which contains $Y$ (\cref{fig:threshold}). An edge $uv$ is cut only if $\tau$ lies between $h(u)$ and $h(v)$, which happens with probability at most $|h(u)-h(v)|/T\le\ell_{uv}/T$. A vertex $v$ lies outside $S_\tau$ with probability $h(v)/T$. Hence
\[
\mathbb{E}\,|\delta_G(S_\tau)|\le\frac1T\sum_e\ell_e\le\frac1T
<\frac{4\phi}{T}\sum_vA(v)h(v)
=4\phi\,\mathbb{E}\,A(V(G)\setminus S_\tau).
\]
So some $\tau$ has $|\delta_G(S_\tau)|<4\phi A(V(G)\setminus S_\tau)$; in particular both sides of the cut have positive weight. Since $A(S_\tau)\ge A(Y)\ge|A|/3$, the complement has weight at most $2|A|/3\le2A(S_\tau)$, so $|\delta_G(S_\tau)|\le8\phi\min\{A(S_\tau),A(V(G)\setminus S_\tau)\}$. Let $S$ be the side of smaller weight. The set $S_\tau$ only changes at the values $h(v)$, so trying these $n$ thresholds finds $S$ deterministically.
\end{proof}
A ball of radius $R$ has diameter at most $2R$, and $\log(|A|/A(S))\ge\log2$, so \cref{lem:heavy-ball} proves \cref{lem:size-cut} whenever some ball of radius $R$ has weight at least $|A|/3$. So we may assume from now on that
\begin{equation}\label{eq:light-balls}
A(B_d(u,R))<|A|/3
\qquad\text{for every vertex }u.
\end{equation}

\paragraph{Volumes.}
For a set $U\subseteq V(G)$, a vertex $u\in U$, and $r>0$, let $B^U(u,r)=B_d(u,r)\cap U$. Its \emph{volume} counts $A(v)/|A|$ for every vertex $v$ of the ball, the length of every edge inside it, and the part inside the ball of every edge leaving it:
\[
\Vol_u^U(r)=\frac{A(B^U(u,r))}{|A|}
 +\sum_{xy\in E(G[B^U(u,r)])}\ell_{xy}
 +\sum_{\substack{xy\in E(G[U])\\x\in B^U(u,r),\ y\notin B^U(u,r)}}
 \bigl(r-d(u,x)\bigr).
\]
We use two deficits, one for the volume of the ball and one for the weight of a set $S$:
\[
F_u^U(r)=\log\frac{2e}{\Vol_u^U(r)},
\qquad
\Gamma(S)=\log\frac{e|A|}{A(S)}.
\]

\begin{lemma}\label{lem:expander-volume}
For every $U\subseteq V(G)$, every $u\in U$, and every $r>0$:
\begin{enumerate}[label=(\roman*),leftmargin=*]
\item $\Vol_u^U$ is a volume function (\cref{sec:prelim-rg}) with boundary cost $|\delta_{G[U]}(B^U(u,r))|$ and normalization $2$;
\item $1\le F_u^U(r)\le\log(2e|A|)$ and $1\le\Gamma(B^U(u,r))\le\log(e|A|)$, and both are nonincreasing in $r$;
\item removing vertices other than $u$ from $U$ decreases neither $F_u^U(r)$ nor $\Gamma(B^U(u,r))$.
\end{enumerate}
\end{lemma}
\begin{proof}
Write $B=B^U(u,r)$.
\begin{enumerate}[label=(\roman*),leftmargin=*]
\item Between two consecutive values of $d(u,\cdot)$, the ball $B$ does not change, and every edge $xy$ of $G[U]$ with $x\in B$ and $y\notin B$ contributes $r-d(u,x)$, which grows with slope one. So the derivative of $\Vol_u^U$ is $|\delta_{G[U]}(B)|$. When $r$ reaches $d(u,y)$ and $y$ enters the ball, the term $A(y)/|A|$ is added, and for every edge $xy$ with $x$ already in the ball, the partial term $d(u,y)-d(u,x)$ is replaced by the full length $\ell_{xy}$; this is an increase, since $d(u,y)-d(u,x)\le\ell_{xy}$ by the triangle inequality. So $\Vol_u^U$ only jumps upward.
\item The vertex terms sum to at most one and the edge terms to at most $\sum_e\ell_e\le1$, so $\Vol_u^U(r)\le2$ and $F_u^U(r)\ge1$. Since $u\in B$, we have $\Vol_u^U(r)\ge A(u)/|A|\ge1/|A|$, so $F_u^U(r)\le\log(2e|A|)$. Similarly, $1\le A(u)\le A(B)\le|A|$ gives $1\le\Gamma(B)\le\log(e|A|)$. Both deficits are nonincreasing in $r$, because $\Vol_u^U(r)$ and $A(B)$ are nondecreasing in $r$.
\item Removing a vertex $y\ne u$ from $U$ removes its own term and the terms of the edges incident to $y$, and leaves all other terms unchanged. So neither $\Vol_u^U(r)$ nor $A(B)$ increases, and neither deficit decreases.
\end{enumerate}
\end{proof}

\paragraph{The algorithm.}
Unlike \cref{sec:cutwidth,sec:directed,sec:vertex}, which use Global Carving, here we use Greedy Seymour Carving (\cref{sec:ov-greedy}): the averaging argument in the proof of \cref{lem:size-cut} needs a bound on the sum of the costs of all pieces, each relative to its own weight, whereas Global Carving only controls the prefixes of the carving.

Let $\kappa=2\lceil\log_2\log(2e|A|)\rceil+1=O(\log\log|A|)$, and divide $[R/2,R]$ into the $\kappa$ annuli $I_j=[a_j,a_{j-1}]$ for $j\in[\kappa]$, where $a_j=R-jR/(2\kappa)$; so $I_1$ is the outermost annulus and $I_\kappa$ the innermost. For a set $U$ and a vertex $u\in U$, the annulus $I_j$ is \emph{stable for $u$ in $U$} if
\[
F_u^U(a_j)\le2F_u^U(a_{j-1})
\qquad\text{and}\qquad
\Gamma(B^U(u,a_j))\le2\,\Gamma(B^U(u,a_{j-1})).
\]
The algorithm for \cref{lem:size-cut}, assuming \eqref{eq:light-balls}, is described in \cref{alg:expander}. At a high level, it maintains a residual set $U$, initially $V(G)$, and removes one ball from $U$ in each iteration until $U$ is empty. Every vertex $u$ keeps an index $j(u)$, which only increases, and its current annulus is $I_{j(u)}$. The \emph{key} of $u$ is the weight $A(B^U(u,a_{j(u)-1}))$ of the ball at the outer end of its current annulus; unlike in \cref{sec:ov-greedy}, keys are weights rather than volumes, since the pieces are compared by weight. Each iteration carves around a vertex of maximum key. At the end, the algorithm outputs the piece whose boundary is smallest relative to its weight.

\begin{framed}
\noindent\refstepcounter{algorithm}\label{alg:expander}\textbf{Algorithm \thealgorithm} (Greedy Seymour Carving, weighted)\\[2pt]
\textbf{Input:} a connected graph $G$ with a node-weighting $A$, and lengths $\ell$ satisfying \eqref{eq:failed-dual} and \eqref{eq:light-balls}.\\
\textbf{Output:} a set $S\subseteq V(G)$ as in \cref{lem:size-cut}.
\begin{enumerate}[label=\arabic*:,leftmargin=2.2em,itemsep=1pt,topsep=4pt]
\item Initialize $U\leftarrow V(G)$, $i\leftarrow0$, and $j(u)\leftarrow1$ for every vertex $u$.
\item \textbf{while} $U\ne\emptyset$ \textbf{do}
\item \quad \textbf{for} every $u\in U$: \textbf{while} $I_{j(u)}$ is not stable for $u$ in $U$, set $j(u)\leftarrow j(u)+1$.
\item \quad $i\leftarrow i+1$, and let $u_i\in U$ be a vertex of maximum key.
\item \quad Apply \cref{lem:stable-annulus} to $\Vol_{u_i}^U$ on $I_{j(u_i)}$, and let $r_i$ be the resulting radius.
\item \quad $X_i\leftarrow B^U(u_i,r_i)$, $U_i\leftarrow U$, and $U\leftarrow U\setminus X_i$.
\item \textbf{return} the piece $X_i$ minimizing $|\delta_G(X_i)|/\bigl(A(X_i)\,\Gamma(X_i)\bigr)$.
\end{enumerate}
\end{framed}

\paragraph{Analysis.}
We first show that \cref{alg:expander} is well defined and that every piece is cheap in the residual graph.

\begin{lemma}\label{lem:expander-well-defined}
Every index $j(u)$ stays in $[\kappa]$, and \cref{alg:expander} terminates after at most $n$ iterations, in polynomial time. For every $i$,
\[
|\delta_{G[U_i]}(X_i)|\le O\Big(\frac{\log\log|A|}{R}\Big)\Vol_{u_i}^{U_i}(r_i)\,F_{u_i}^{U_i}(r_i)
\qquad\text{and}\qquad
1\le A(X_i)<|A|/3.
\]
\end{lemma}
\begin{proof}
Fix a vertex $u$. Whenever $j(u)$ increases, the annulus $I_{j(u)}$ is unstable for $F_u^U$ or for $\Gamma(B^U(u,\cdot))$, so the move is forced by one of them in the sense of \cref{lem:persistence}. By \cref{lem:expander-volume}, that lemma applies to both, with values in $[1,\log(2e|A|)]$, so each forces at most $\log_2\log(2e|A|)$ increases, and $2\log_2\log(2e|A|)<\kappa$. Each iteration removes $u_i\in X_i$, so there are at most $n$ iterations.

Line~5 applies \cref{lem:stable-annulus} to an annulus of width $\Omega(R/\log\log|A|)$ that is stable for $F_{u_i}^{U_i}$, which gives the bound on $|\delta_{G[U_i]}(X_i)|$. Finally, $A(X_i)\ge A(u_i)\ge1$, and $A(X_i)<|A|/3$ by \eqref{eq:light-balls}, since $X_i\subseteq B_d(u_i,R)$.
\end{proof}

The choice of a maximum key in Line~4 makes the weight deficits of the pieces nondecreasing up to a factor of two.

\begin{lemma}\label{lem:expander-order}
$\Gamma(X_i)\le2\,\Gamma(X_j)$ whenever $i<j$.
\end{lemma}
\begin{proof}
Removing vertices and increasing indices can only decrease a key, so the key of every remaining vertex only decreases over time. Let $Q_i$ be the key of $u_i$ when it is selected. Since Line~4 always selects a maximum key, $Q_1\ge Q_2\ge\cdots$. If the annulus of $u_i$ is $I_p$, then $a_p<r_i$ gives $A(B^{U_i}(u_i,a_p))\le A(X_i)\le Q_i$, and stability gives
\[
\log\frac{e|A|}{Q_i}\le\Gamma(X_i)\le\Gamma(B^{U_i}(u_i,a_p))\le2\log\frac{e|A|}{Q_i}.
\]
For $i<j$, we get $\Gamma(X_i)\le2\log\frac{e|A|}{Q_i}\le2\log\frac{e|A|}{Q_j}\le2\,\Gamma(X_j)$.
\end{proof}

Next, the volumes of the pieces are bounded in total, and their boundaries in $G$ are comparable to their boundaries in the residual graph.

\begin{lemma}\label{lem:expander-carving}
$\sum_i\Vol_{u_i}^{U_i}(r_i)\le2$, $\sum_iA(X_i)=|A|$, and
\[
\sum_i\frac{|\delta_G(X_i)|}{\Gamma(X_i)}\le3\sum_i\frac{|\delta_{G[U_i]}(X_i)|}{\Gamma(X_i)}.
\]
\end{lemma}
\begin{proof}
Every vertex $v$ contributes $A(v)/|A|$ only to the volume of its own piece, and every edge contributes at most its length, and only to the volume of the piece containing the endpoint removed first. So $\sum_i\Vol_{u_i}^{U_i}(r_i)\le1+\sum_e\ell_e\le2$. The pieces partition $V(G)$, so $\sum_iA(X_i)=|A|$.

For the last bound, an edge between $X_i$ and $X_j$ with $i<j$ lies in $\delta_{G[U_i]}(X_i)$. It contributes $1/\Gamma(X_i)$ to the right-hand sum, and $1/\Gamma(X_i)+1/\Gamma(X_j)\le3/\Gamma(X_i)$ to the left-hand sum, by \cref{lem:expander-order}.
\end{proof}

We can now prove the main lemma.

\begin{proof}[Proof of \cref{lem:size-cut}]
If some ball of radius $R$ has weight at least $|A|/3$, we are done by \cref{lem:heavy-ball}. Otherwise, run \cref{alg:expander}. We use the elementary inequality
\begin{equation}\label{eq:entropy-comparison}
\frac{p\log(e/p)}{\log(e/q)}\le p+\frac qe
\qquad(0<p,q\le1).
\end{equation}
If $p\ge q$, the left-hand side is at most $p$. Otherwise write $p=xq$ with $0<x<1$; then the left-hand side is $p+xq\log(1/x)/\log(e/q)\le p+q/e$, since $x\log(1/x)\le1/e$.

Apply it with $p=\Vol_{u_i}^{U_i}(r_i)/2$ and $q=A(X_i)/|A|$, for which $\Vol_{u_i}^{U_i}(r_i)F_{u_i}^{U_i}(r_i)=2p\log(e/p)$ and $\Gamma(X_i)=\log(e/q)$. By \cref{lem:expander-carving,lem:expander-well-defined},
\[
\begin{aligned}
\sum_i\frac{|\delta_G(X_i)|}{\Gamma(X_i)}
&\le O\Big(\frac{\log\log|A|}{R}\Big)\sum_i\frac{\Vol_{u_i}^{U_i}(r_i)\,F_{u_i}^{U_i}(r_i)}{\Gamma(X_i)}\\
&\le O\Big(\frac{\log\log|A|}{R}\Big)\sum_i\Big(\Vol_{u_i}^{U_i}(r_i)+\frac{2A(X_i)}{e|A|}\Big)
=O\Big(\frac{\log\log|A|}{R}\Big).
\end{aligned}
\]
Since $\sum_iA(X_i)=|A|$, some piece satisfies $\frac{|\delta_G(X_i)|}{A(X_i)\Gamma(X_i)}\le\frac1{|A|}\sum_j\frac{|\delta_G(X_j)|}{\Gamma(X_j)}$: otherwise, multiplying each ratio by $A(X_i)$ and summing over $i$ would give more than $\sum_i\frac{|\delta_G(X_i)|}{\Gamma(X_i)}$. The piece $X_i$ returned in Line~7 minimizes this ratio, so it satisfies
\[
|\delta_G(X_i)|\le O\Big(\frac{\log\log|A|}{R|A|}\Big)A(X_i)\Gamma(X_i)=O(\phi\log\log|A|)\,A(X_i)\log\frac{|A|}{A(X_i)},
\]
where the last step uses $A(X_i)\le|A|/2$. Line~7 can be implemented with rational arithmetic, up to a constant factor (\cref{app:implementation}).
\end{proof}

\subsection{The decomposition}
The decomposition algorithm is recursive. On a connected instance $G'$ (initially each connected component of the input graph), compute the minimum congestion for routing $\Pi$ in $G'$. If it is at most $1/(2\phi)$, stop: $A$ is $\phi$-flow-expanding in $G'$ by \cref{lem:product-certificate}. Otherwise, find the set $S$ of \cref{lem:size-cut}, delete $\delta_{G'}(S)$, and recurse on the connected components of $G'[S]$ and $G'[V(G')\setminus S]$.

\begin{proof}[Proof of \cref{thm:weighted-expander}]
Every remaining component passes the test, so $A$ is $\phi$-flow-expanding in it. Every cut splits an instance into two nonempty parts, so there are fewer than $2n$ instances, each handled in polynomial time. It remains to count the deleted edges.

By \cref{lem:size-cut}, when an instance $G'$ of weight $M=A(V(G'))$ is cut along $S$, the number of deleted edges is at most $C\phi\log\log|A|\cdot A(S)\log\frac{M}{A(S)}$ for an absolute constant $C$. We charge this amount to the vertices of $S$, each vertex $v\in S$ receiving $C\phi\log\log|A|\cdot A(v)\log\frac{M}{A(S)}$. Fix a vertex $v$, and let $M_1>M_2>\cdots>M_k$ be the weights of the instances in which $v$ is charged, with $S_t$ the side containing $v$ in the $t$-th of them. Every later instance containing $v$ lies inside $S_t$, so $M_{t+1}\le A(S_t)$. Since $A(S_k)\ge1$,
\[
\sum_{t=1}^k\log\frac{M_t}{A(S_t)}
\le\sum_{t<k}\log\frac{M_t}{M_{t+1}}+\log M_k
=\log M_1\le\log|A|.
\]
So every vertex $v$ is charged $O(\phi\log|A|\log\log|A|)A(v)$ in total, and summing over all vertices bounds the number of deleted edges by $O(\phi|A|\log|A|\log\log|A|)$.
\end{proof}

\section{Directed Layouts: Minimum Storage}\label{sec:directed}
Let $G=(V,E)$ be a DAG. For a topological ordering $\pi$, let $\MS_\pi(G)=\max_i|\delta_G(S_i^\pi)|$ be its \emph{storage}, and let $\MS(G)=\min_\pi\MS_\pi(G)$. Every prefix of a topological ordering is closed under predecessors, so all arcs crossing a prefix leave it.

\begin{theorem}\label{thm:minimum-storage}
There is a polynomial-time algorithm that, given an $n$-vertex DAG $G$, computes a topological ordering of storage at most $O(\log n)\MS(G)$.
\end{theorem}

\subsection{The relaxation}
We use the following LP~\cite{ENRS00,RR05}:
\begin{equation}\label{eq:directed-lp}
\begin{aligned}
\DLP(G)=\min\quad &\sum_{xy\in E}d(x,y)\\
\text{s.t.}\quad &d(x,z)\le d(x,y)+d(y,z) &&x,y,z\in V,\\
&d(x,x)=0,\quad d(x,y)\ge0 &&x,y\in V,\\
&d(y,x)=0 &&xy\in E,\\
&\sum_{v\in S}\bigl(d(u,v)+d(v,u)\bigr)
 \ge\frac{|S|^2-1}{4} &&u\in V,\ S\subseteq V.
\end{aligned}
\end{equation}
The constraints $d(y,x)=0$ encode precedence, and the last family imposes spreading on the symmetrization $\bar d(x,y)=d(x,y)+d(y,x)$, which is a metric.

\begin{lemma}\label{lem:directed-lower}
The LP \eqref{eq:directed-lp} can be solved in polynomial time, and $\DLP(G)\le |V(G)|\MS(G)$.
\end{lemma}
\begin{proof}
For a topological ordering $\pi$, set $d_\pi(x,y)=[\pi(y)-\pi(x)]_+$. It satisfies the directed triangle inequality, $d_\pi(y,x)=0$ for every arc $xy$, and $\bar d_\pi(x,y)=|\pi(x)-\pi(y)|$, which satisfies spreading by \cref{lem:line-spreading}; so it is feasible. Its objective is $\sum_{i<n}|\delta_G(S_i^\pi)|\le n\MS_\pi(G)$. The spreading constraints are separated as in \cref{app:implementation}, using the values $\bar d(u,v)$.
\end{proof}
Restricting a solution to an induced subgraph preserves feasibility. Moreover, along any path of the underlying undirected graph, each arc contributes its length to $\bar d$, so $\bar d(u,v)\le\DLP(G)$ for $u,v$ in the same weak component. As in the proof of \cref{lem:wlp-to-cw}, it follows that $\DLP(G)\ge3/8$ for every weakly connected $G$ with at least two vertices.

\subsection{Source balls, sink balls, and volumes}
For a vertex $u$ and $r>0$, the \emph{source ball} and the \emph{sink ball} around $u$ are
\[
B^+(u,r)=\{v:d(u,v)<r\},
\qquad
B^-(u,r)=\{v:d(v,u)<r\}.
\]
The names come from the following property: in any induced subgraph, a source ball only has outgoing arcs, and a sink ball only has incoming arcs.

\begin{lemma}\label{lem:source-sink}
Let $U\subseteq V$, $u\in U$, and $r>0$. Every arc of $G[U]$ with exactly one endpoint in $B^+(u,r)\cap U$ leaves $B^+(u,r)\cap U$, and every arc of $G[U]$ with exactly one endpoint in $B^-(u,r)\cap U$ enters $B^-(u,r)\cap U$.
\end{lemma}
\begin{proof}
If $xy\in E$, then $d(y,x)=0$, so $d(u,x)\le d(u,y)$ and $d(y,u)\le d(x,u)$. Hence $y\in B^+(u,r)$ implies $x\in B^+(u,r)$, so no arc enters a source ball; and $x\in B^-(u,r)$ implies $y\in B^-(u,r)$, so no arc leaves a sink ball.
\end{proof}

From now on, $G$ is weakly connected with $n\ge2$ vertices and $d$ is an optimal solution of \eqref{eq:directed-lp}. For every vertex $u$, one of $B^+(u,n/20)$ and $B^-(u,n/20)$ has at most $3n/4$ vertices: their intersection lies in the $\bar d$-ball of radius $n/10$, which has at most $n/2$ vertices by the argument of \cref{lem:cw-volume}(iv), applied to $\bar d$, on which \eqref{eq:directed-lp} imposes spreading. We fix such a sign $\sigma(u)\in\{+,-\}$ for every vertex; we say that $u$ is of \emph{source type} if $\sigma(u)=+$ and of \emph{sink type} otherwise, and write
\[
B^U(u,r)=B^{\sigma(u)}(u,r)\cap U
\]
for its \emph{oriented ball} in a set $U\ni u$. By \cref{lem:source-sink}, every arc of $G[U]$ crossing $B^U(u,r)$ leaves it if $u$ is of source type, and enters it if $u$ is of sink type.

As in \cref{sec:cutwidth}, the \emph{volume} of a set $S\subseteq V$ is $\Vol(S)=\frac{\DLP(G)}{n}|S|+\sum_{xy\in E(G[S])}d(x,y)$. The volume of an oriented ball adds the part inside the ball of every arc of $G[U]$ crossing it:
\[
\Vol_u^U(r)=\Vol(B^U(u,r))+
\begin{cases}
\sum_{xy}\bigl(r-d(u,x)\bigr) & \text{if $\sigma(u)=+$, over arcs $xy$ of $G[U]$ leaving the ball},\\[2pt]
\sum_{xy}\bigl(r-d(y,u)\bigr) & \text{if $\sigma(u)=-$, over arcs $xy$ of $G[U]$ entering the ball}.
\end{cases}
\]
For $r\le0$ the ball is empty and its volume is $0$.

\begin{lemma}\label{lem:directed-volume}
For every $U\subseteq V$, every $u\in U$, and every $r>0$:
\begin{enumerate}[label=(\roman*),leftmargin=*]
\item $\Vol_u^U$ is a volume function (\cref{sec:prelim-rg}) with boundary cost $|\delta_{G[U]}(B^U(u,r))|$ and normalization $2\DLP(G)$; its events are the values $d(u,x)$ if $u$ is of source type, and $d(x,u)$ otherwise;
\item $\DLP(G)/n\le\Vol_u^U(r)\le2\DLP(G)$;
\item removing vertices other than $u$ from $U$ does not increase $\Vol_u^U(r)$;
\item if $r\le n/20$, then $|B^U(u,r)|\le3n/4$.
\end{enumerate}
\end{lemma}
\begin{proof}
Write $B=B^U(u,r)$. We give the argument for $u$ of source type; for $u$ of sink type, replace $d(u,\cdot)$ by $d(\cdot,u)$ and leaving arcs by entering arcs.
\begin{enumerate}[label=(\roman*),leftmargin=*]
\item Between two consecutive values of $d(u,\cdot)$, the ball $B$ does not change, and every arc $xy$ of $G[U]$ leaving $B$ contributes $r-d(u,x)$, which grows with slope one. By \cref{lem:source-sink}, these are all the arcs of $G[U]$ crossing $B$, so the derivative of $\Vol_u^U$ is $|\delta_{G[U]}(B)|$. When $r$ reaches $d(u,y)$ and $y$ enters the ball, the term $\DLP(G)/n$ of $y$ is added, and for every arc $xy$ with $x$ already in the ball, the partial term $d(u,y)-d(u,x)$ is replaced by the full length $d(x,y)$; this is an increase, since $d(u,y)-d(u,x)\le d(x,y)$ by the triangle inequality. So $\Vol_u^U$ only jumps upward.
\item The vertex terms sum to at most $\DLP(G)$, and the arc terms to at most $\sum_{xy\in E}d(x,y)=\DLP(G)$. Since $u\in B$, the volume is at least $\DLP(G)/n$.
\item Removing a vertex $y\ne u$ from $U$ removes its own term and the terms of the arcs incident to $y$, and leaves all other terms unchanged.
\item This holds by the choice of $\sigma(u)$.
\end{enumerate}
\end{proof}

\subsection{The algorithm}
\Cref{alg:directed} is \cref{alg:cutwidth} with oriented balls, followed by a reordering of the pieces.

\begin{framed}
\noindent\refstepcounter{algorithm}\label{alg:directed}\textbf{Algorithm \thealgorithm} (Global Carving, directed)\\[2pt]
\textbf{Input:} a weakly connected DAG $G$ with $n\ge2$ vertices, and an optimal solution $d$ of \eqref{eq:directed-lp}.\\
\textbf{Output:} an ordered partition of $V(G)$.
\begin{enumerate}[label=\arabic*:,leftmargin=2.2em,itemsep=1pt,topsep=4pt]
\item Initialize $U\leftarrow V(G)$, $\rho\leftarrow n/20$, and $i\leftarrow0$.
\item \textbf{while} $U\ne\emptyset$ \textbf{do}
\item \quad $i\leftarrow i+1$, $U_i\leftarrow U$, $\rho_i\leftarrow\rho$, $K_i\leftarrow\max_{u\in U}\Vol_u^U(\rho)$, and $w_i\leftarrow\frac{n}{320}\cdot\frac{K_i}{\DLP(G)}$.
\item \quad Let $u_i\in U$ be a vertex maximizing $\Vol_u^U(\rho-w_i)$.
\item \quad Let $r_i\in(\rho-w_i,\rho)$ be a radius that is not an event of $\Vol_{u_i}^U$ and minimizes $|\delta_{G[U]}(B^U(u_i,r_i))|$.
\item \quad $X_i\leftarrow B^U(u_i,r_i)$, $U\leftarrow U\setminus X_i$, and $\rho\leftarrow\rho-w_i$.
\item \textbf{return} the pieces $X_1,\ldots,X_k$ in the two-ended order defined below.
\end{enumerate}
\end{framed}

\noindent As in \cref{sec:cw-algorithm}, Line~5 is well defined as long as $\rho>0$; we show in \cref{lem:directed-well-defined} that $\rho$ stays above $n/40$. The window of iteration $i$ is $(\rho_{i+1},\rho_i)$, where $\rho_{i+1}=\rho_i-w_i$, and we write $W_i=\Vol_{u_i}^{U_i}(r_i)$ for the volume of $X_i$ when it is carved. The index $i$ gives the \emph{removal order} of the pieces. Call $X_i$ a \emph{source piece} if $u_i$ is of source type, and a \emph{sink piece} otherwise. The \emph{two-ended order} lists the source pieces in removal order, followed by the sink pieces in reverse removal order.

\begin{lemma}\label{lem:two-ended}
Let $i<j$. Every arc between $X_i$ and $X_j$ goes from the piece that comes earlier in the two-ended order to the piece that comes later. In particular, any ordering of $V$ that lists the pieces in the two-ended order, each piece as a contiguous block ordered topologically, is a topological ordering of $G$.
\end{lemma}
\begin{proof}
When $X_i$ is removed, the remaining graph is $G[U_i]$, and all later pieces lie in $U_i\setminus X_i$. By \cref{lem:source-sink}, a source ball has only outgoing arcs in the remaining graph, and a sink ball has only incoming arcs. So every arc between $X_i$ and a later piece leaves $X_i$ if $X_i$ is a source piece, and enters $X_i$ if $X_i$ is a sink piece. The two-ended order respects this. It places every source piece before all pieces removed after it, and every sink piece after them. Hence every arc between two pieces points forward. Arcs inside a piece point forward by assumption, which gives the last claim.
\end{proof}

\paragraph{Analysis in removal order.}
The next three lemmas are the analogues of \cref{lem:heaviest,lem:carving-props,lem:cw-prefix}, and their proofs are the same.

\begin{lemma}[Volumes]\label{lem:directed-heaviest}
For every iteration $i$:
\begin{enumerate}[label=(\roman*),leftmargin=*]
\item $\Vol(X_i)\le W_i\le K_i$;
\item if $i<k$, then $K_{i+1}\le\Vol_{u_i}^{U_i}(\rho_{i+1})\le W_i$, and hence $W_{i+1}\le W_i$.
\end{enumerate}
Moreover, $\sum_iW_i\le2\DLP(G)$.
\end{lemma}
\begin{proof}
As in the proof of \cref{lem:heaviest}, with arcs in place of edges, using \cref{lem:directed-volume}(iii) and the choice of $u_i$ in Line~4.
\end{proof}

\begin{lemma}[Radius]\label{lem:directed-well-defined}
The radius never drops below $n/40$, that is, $\rho_{i+1}\ge n/40$ for every $i$. Consequently, \cref{alg:directed} terminates after at most $n$ iterations, in polynomial time, and $|X_i|\le3n/4$ for every $i$.
\end{lemma}
\begin{proof}
As in the proof of \cref{lem:carving-props}, \cref{lem:directed-heaviest}(ii) gives $K_j\le2W_j+2(K_j-K_{j+1})$ for $j<k$. So $\sum_{j\le i}K_j\le2\sum_{j<i}W_j+2K_1\le8\DLP(G)$, and $\rho_{i+1}=\frac n{20}-\frac{n}{320\,\DLP(G)}\sum_{j\le i}K_j\ge\frac n{40}$. Since $r_i<n/20$, \cref{lem:directed-volume}(iv) gives $|X_i|\le3n/4$, and since $r_i>0$, every iteration removes $u_i$.
\end{proof}

\begin{lemma}[Prefix bound]\label{lem:directed-carving}
For every $i$,
\begin{equation}\label{eq:directed-prefix}
\sum_{j\le i}|\delta_G(X_j)|\le640\,\frac{\DLP(G)}{n}\log\frac{2e\,\DLP(G)}{W_i}.
\end{equation}
\end{lemma}
\begin{proof}
Let $L_j=\log\frac{2e\,\DLP(G)}{K_j}\ge1$ be the level at the start of iteration $j$, and let $\Delta_j=\log\frac{\Vol_{u_j}^{U_j}(\rho_j)}{\Vol_{u_j}^{U_j}(\rho_{j+1})}$, which is finite since $\rho_{j+1}>0$. As in the proof of \cref{lem:cw-prefix}, averaging over the window, with the boundary as the derivative of $\Vol_{u_j}^{U_j}$ (\cref{lem:directed-volume}(i)), gives
\[
|\delta_{G[U_j]}(X_j)|\le\frac{K_j}{w_j}\min\{1,\Delta_j\}=320\,\frac{\DLP(G)}{n}\min\{1,\Delta_j\},
\]
and \cref{lem:directed-heaviest}(ii) gives $L_{j+1}\ge L_j+\Delta_j$ for $j<k$. Summing over $j<i$, and adding at most $320\frac{\DLP(G)}{n}\le320\frac{\DLP(G)}{n}L_1$ for iteration $i$, we get $\sum_{j\le i}|\delta_{G[U_j]}(X_j)|\le320\frac{\DLP(G)}{n}L_i$. An arc between $X_j$ and $X_{j'}$ with $j<j'$ lies in $\delta_{G[U_j]}(X_j)$, and it is counted at most twice on the left-hand side of \eqref{eq:directed-prefix}, and only if $j\le i$. Hence the left-hand side is at most $640\frac{\DLP(G)}{n}L_i\le640\frac{\DLP(G)}{n}\log\frac{2e\,\DLP(G)}{W_i}$, since $W_i\le K_i$.
\end{proof}

\paragraph{The two-ended order.}
It remains to relate prefixes of the two-ended order to the removal order.

\begin{lemma}\label{lem:directed-blocks}
Let $P\subseteq V$ be a prefix of an ordering of $V$ that lists the pieces in the two-ended order, each piece as a contiguous block, and suppose that the last vertex of $P$ lies in $X_i$. Then every arc of $\delta_G(P)$ that is not an arc of $G[X_i]$ lies in $\delta_G(X_j)$ for some $j\le i$; in particular, there are at most $\sum_{j\le i}|\delta_G(X_j)|$ such arcs.
\end{lemma}
\begin{proof}
Consider an arc $e\in\delta_G(P)$ that is not an arc of $G[X_i]$. If $e$ has an endpoint in $X_i$, then its other endpoint is outside $X_i$, so $e\in\delta_G(X_i)$. Otherwise, suppose first that $X_i$ is a source piece. Then $P\setminus X_i$ is the union of the source pieces $X_j$ with $j<i$, and $e$ has an endpoint in $P\setminus X_i$, so $e\in\delta_G(X_j)$ for such a $j$. If $X_i$ is a sink piece, then $V\setminus(P\cup X_i)$ is the union of the sink pieces $X_j$ with $j<i$, and $e$ has an endpoint there, so again $e\in\delta_G(X_j)$ for such a $j$.
\end{proof}

\subsection{Proof of \texorpdfstring{\cref{thm:minimum-storage}}{the theorem}}
The final ordering is computed recursively. Process weak components separately; a single vertex is ordered trivially. On a weakly connected instance with at least two vertices, solve \eqref{eq:directed-lp} and run \cref{alg:directed}; then recursively order the weak components of each piece, and concatenate these orders following the two-ended order of the pieces. By \cref{lem:two-ended} and induction, the result is a topological ordering.

\begin{proof}[Proof of \cref{thm:minimum-storage}]
There are no arcs between weak components, so we may assume that $G$ is weakly connected with at least two vertices. Fix a prefix $P$ of the final ordering, and let $v$ be its last vertex. Let $H_0=G,H_1,\ldots,H_T$ be the recursive instances with at least two vertices that contain $v$. For $t\le T$, let $Y_t$ be the piece of $H_t$ containing $v$; $H_{t+1}$ is the weak component of $H_t[Y_t]$ containing $v$, where $H_{T+1}=\{v\}$. Write $\Vol_{H_t}$ for the volume computed in $H_t$.

Every arc of $\delta_G(P)$ is charged to the last level $t\le T$ at which both of its endpoints lie in $H_t$. At level $t$, the charged arcs are arcs of $\delta_G(P)$ inside $H_t$ that are not arcs of $H_t[Y_t]$; an arc of $H_t[Y_t]$ outside $H_{t+1}$ does not cross $P$, since it lies in another weak component of $H_t[Y_t]$, which is ordered as a contiguous block not containing $v$, and hence lies entirely inside or entirely outside $P$. Since $P\cap V(H_t)$ is a prefix of the ordering of $H_t$, \cref{lem:directed-blocks,lem:directed-carving} bound the number of arcs charged to level $t$ by
\[
640\,\frac{\DLP(H_t)}{|V(H_t)|}\log\frac{2e\,\DLP(H_t)}{\Vol_{H_t}(Y_t)}
\le640\,\MS(G)\log\frac{2e\,\DLP(H_t)}{\Vol_{H_t}(Y_t)}.
\]
Here we used that the volume of $Y_t$ is at most its volume when it is carved (\cref{lem:directed-heaviest}(i)), and that $\DLP(H_t)\le|V(H_t)|\MS(H_t)\le|V(H_t)|\MS(G)$ by \cref{lem:directed-lower}.

It remains to bound the sum of the logarithms. Since restricting the solution of $H_t$ to $H_{t+1}$ is feasible, $\DLP(H_{t+1})\le\Vol_{H_t}(Y_t)$ for $t<T$, and $\Vol_{H_T}(Y_T)\ge\DLP(H_T)/|V(H_T)|$, since every vertex contributes $\DLP(H_T)/|V(H_T)|$ to the volume. Hence
\[
\sum_{t=0}^T\log\frac{2e\,\DLP(H_t)}{\Vol_{H_t}(Y_t)}
\le T\log(2e)+\log\frac{\DLP(H_0)}{\DLP(H_T)}+\log\bigl(2e|V(H_T)|\bigr)=O(\log n),
\]
since $T=O(\log n)$ by \cref{lem:directed-well-defined}, $\DLP(G)\le n\MS(G)\le n^3$, and $\DLP(H_T)\ge3/8$. So $|\delta_G(P)|=O(\log n)\MS(G)$.
\end{proof}
\Cref{thm:minimum-storage} is the \textsc{Minimum Storage} statement of \cref{thm:main-directed}.

\begin{remark}[Weighted arcs]\label{rem:weighted-storage}
For positive arc weights $w$, the \emph{weighted storage} of a topological ordering $\pi$ is $\MS^w_\pi(G)=\max_iw(\delta_G(S_i^\pi))$, and $\MS_w(G)=\min_\pi\MS^w_\pi(G)$. \Cref{thm:minimum-storage} also holds for $\MS_w$, for positive rational arc weights $w$ whose largest-to-smallest ratio is polynomial in $n$.
\end{remark}
\begin{proof}
Scale $w$ so that its minimum is $1$, and let $w'_{xy}=\lceil w_{xy}\rceil$. Then $\MS^w_\pi(G)\le\MS^{w'}_\pi(G)\le2\MS^w_\pi(G)$ for every $\pi$. Let $G''$ be obtained from $G$ by replacing every arc $xy$ by $w'_{xy}$ paths $x\to z\to y$, each through a new vertex $z$. Let $\sigma$ be a topological ordering of $G''$, and let $\pi$ be its restriction to $V$. Each such $z$ lies between $x$ and $y$ in $\sigma$, so a path crosses a prefix of $\sigma$ if and only if $x$ lies in the prefix and $y$ does not, and then it crosses it exactly once. Hence $\MS_\sigma(G'')=\MS^{w'}_\pi(G)$. Conversely, every $\pi$ extends to such a $\sigma$ by placing each $z$ right after $x$. Since $G''$ has polynomially many vertices, \cref{thm:minimum-storage} applied to $G''$ proves the claim.
\end{proof}

\section{Pathwidth and Treedepth}\label{sec:vertex}
This section proves \cref{thm:main-treedepth} and the \textsc{Pathwidth} statement of \cref{thm:main-layout}.

\subsection{The relaxation}
For a graph $G$ with $n$ vertices, we use the vertex-separation relaxation of~\cite{BKS24}:
\begin{equation}\label{eq:vertex-lp}
\begin{aligned}
\LPVS(G)=\min\quad &\sum_{v\in V(G)}x_v\\
\text{s.t.}\quad &x_u+x_v\ge d(u,v) &&uv\in E(G),\\
&d(u,w)\le d(u,v)+d(v,w) &&u,v,w\in V(G),\\
&\sum_{v\in S}d(u,v)\ge\frac{|S|^2-1}{4} &&u\in V(G),\ S\subseteq V(G),\\
&d(u,v)=d(v,u),\quad d(u,u)=0 &&u,v\in V(G),\\
&d(u,v)\ge0,\quad x_v\ge0 &&u,v\in V(G).
\end{aligned}
\end{equation}
It has a polynomial-time separation oracle (\cref{app:implementation}).

\begin{lemma}\label{lem:vertex-lower}
For every graph $H$, $\LPVS(H)\le |V(H)|\pw(H)$.
\end{lemma}
\begin{proof}
By \cref{lem:vs-pw}, some ordering $\pi$ has vertex-separation cost $\pw(H)$. Set $d(u,v)=|\pi(u)-\pi(v)|$ and $x_u=\max(\{\pi(v)-\pi(u):uv\in E(H)\}\cup\{0\})$. This is a metric, and it satisfies spreading by \cref{lem:line-spreading}. For every edge $uv$ with $\pi(u)<\pi(v)$, we have $x_u+x_v\ge x_u\ge d(u,v)$. Finally, $x_u$ is the number of prefixes of $\pi$ in which $u$ is active, that is, in which $u$ is processed and has an unprocessed neighbor. So $\sum_ux_u$ is the sum over all prefixes of the number of active vertices, which is at most $|V(H)|\pw(H)$.
\end{proof}

\begin{lemma}\label{lem:forest-order}
Given an elimination forest of height $h$, a depth-first preorder has vertex-separation cost at most $h-1$. In particular, $\pw(H)\le\td(H)-1$ for every graph $H$.
\end{lemma}
\begin{proof}
Consider a proper prefix of the preorder, and let $z$ be the next vertex. If a processed vertex $u$ has an unprocessed neighbor $v$, then $u$ and $v$ are comparable, so $u$ is an ancestor of $v$. At a depth-first boundary, every processed vertex with an unprocessed descendant is a strict ancestor of $z$, and there are at most $h-1$ of them. The last claim follows from \cref{lem:vs-pw}.
\end{proof}

\subsection{Cores and separators}
From now on, $G$ is connected with $n\ge2$ vertices and $(d,x)$ is an optimal solution of \eqref{eq:vertex-lp}. Around a vertex $u$, each vertex $v$ occupies the interval of radii $[d(u,v)-x_v,\,d(u,v)+x_v]$. For a set $U\subseteq V(G)$, a vertex $u\in U$, and a radius $r$, the \emph{core} consists of the vertices of $U$ whose interval lies below $r$, and the \emph{separator} of those whose interval contains $r$:
\[
S^U(u,r)=\{v\in U:d(u,v)+x_v<r\},
\qquad
D^U(u,r)=\{v\in U:d(u,v)-x_v<r<d(u,v)+x_v\}.
\]
We only use radii that differ from all values $d(u,v)$ and $d(u,v)\pm x_v$. For such a radius, every $v\in U$ with $d(u,v)-x_v<r$ lies in $S^U(u,r)\cup D^U(u,r)$, and every other $v\in U$ satisfies $d(u,v)-x_v>r$.

\begin{lemma}[Core--separator property]\label{lem:core-separator}
No edge joins $S^U(u,r)$ to $U\setminus(S^U(u,r)\cup D^U(u,r))$. If $r\le n/10$, then $|S^U(u,r)|\le n/2$.
\end{lemma}
\begin{proof}
Let $p\in S^U(u,r)$ and $q\in U\setminus(S^U(u,r)\cup D^U(u,r))$. Then $d(u,p)+x_p<r<d(u,q)-x_q$, so $d(u,q)-d(u,p)>x_p+x_q$. If $pq$ were an edge, the edge constraint would give $x_p+x_q\ge d(p,q)$, and hence $d(u,q)>d(u,p)+d(p,q)$, contradicting the triangle inequality.

For the second claim, every vertex $v$ of the core satisfies $d(u,v)\le d(u,v)+x_v<r\le n/10$. Suppose that the core contains a set $X$ of size $s=\lfloor n/2\rfloor+1$. Then $s\ge2$ and $n\le2s-1$, so $\sum_{v\in X}d(u,v)<sn/10\le s(2s-1)/10\le(s^2-1)/4$, contradicting the spreading constraint for $u$ and $X$.
\end{proof}
So the separator disconnects the core from the rest of $U$, and it plays the role of the edge boundary in \cref{sec:cutwidth}.

\paragraph{Volumes.}
Let $B^U(u,r)=\{v\in U:d(u,v)<r\}$. The volume counts $\LPVS(G)/n$ for every vertex of this ball, and the part below $r$ of every interval:
\[
\Vol_u^U(r)=\frac{\LPVS(G)}{n}\bigl|B^U(u,r)\bigr|
+\sum_{v\in U}\min\bigl\{2x_v,\max\{0,r-d(u,v)+x_v\}\bigr\}.
\]
This formula makes sense for every real $r$.

\begin{lemma}\label{lem:vertex-volume}
For every $U\subseteq V(G)$, every $u\in U$, and every $r>0$:
\begin{enumerate}[label=(\roman*),leftmargin=*]
\item $\Vol_u^U$ is a volume function (\cref{sec:prelim-rg}) with boundary cost $|D^U(u,r)|$ and normalization $3\LPVS(G)$; its events are the values $d(u,v)$ and $d(u,v)\pm x_v$;
\item $\LPVS(G)/n\le\Vol_u^U(r)\le3\LPVS(G)$;
\item removing vertices other than $u$ from $U$ does not increase $\Vol_u^U(r)$.
\end{enumerate}
\end{lemma}
\begin{proof}
\begin{enumerate}[label=(\roman*),leftmargin=*]
\item The term $\min\{2x_v,\max\{0,r-d(u,v)+x_v\}\}$ of a vertex $v$ is continuous in $r$, and away from the values $d(u,v)\pm x_v$ it has slope one exactly when $d(u,v)-x_v<r<d(u,v)+x_v$, that is, when $v\in D^U(u,r)$. So the derivative of $\Vol_u^U$ is $|D^U(u,r)|$. The only jumps come from the term $\LPVS(G)/n$ of a vertex $v$, which is added when $r$ reaches $d(u,v)$; so $\Vol_u^U$ only jumps upward.
\item The vertex terms sum to at most $\LPVS(G)$, and the interval terms to at most $\sum_v2x_v=2\LPVS(G)$. Since $u\in B^U(u,r)$, the volume is at least $\LPVS(G)/n$.
\item Removing a vertex $v\ne u$ from $U$ removes its two terms and leaves all other terms unchanged.
\end{enumerate}
\end{proof}

\subsection{The algorithm}
\Cref{alg:vertex} is \cref{alg:cutwidth} with these volumes, except that each iteration removes a core together with its separator.

\begin{framed}
\noindent\refstepcounter{algorithm}\label{alg:vertex}\textbf{Algorithm \thealgorithm} (Global Carving with separators)\\[2pt]
\textbf{Input:} a connected graph $G$ with $n\ge2$ vertices, and an optimal solution $(d,x)$ of \eqref{eq:vertex-lp}.\\
\textbf{Output:} pairs $(S_1,D_1),\ldots,(S_k,D_k)$ of cores and separators.
\begin{enumerate}[label=\arabic*:,leftmargin=2.2em,itemsep=1pt,topsep=4pt]
\item Initialize $U\leftarrow V(G)$, $\rho\leftarrow n/10$, and $i\leftarrow0$.
\item \textbf{while} $U\ne\emptyset$ \textbf{do}
\item \quad $i\leftarrow i+1$, $U_i\leftarrow U$, $\rho_i\leftarrow\rho$, $K_i\leftarrow\max_{u\in U}\Vol_u^U(\rho)$, and $w_i\leftarrow\frac{n}{240}\cdot\frac{K_i}{\LPVS(G)}$.
\item \quad Let $u_i\in U$ be a vertex maximizing $\Vol_u^U(\rho-w_i)$.
\item \quad Let $r_i\in(\rho-w_i,\rho)$ be a radius that is not an event of $\Vol_{u_i}^U$ and minimizes $|D^U(u_i,r_i)|$.
\item \quad $S_i\leftarrow S^U(u_i,r_i)$, $D_i\leftarrow D^U(u_i,r_i)$, $U\leftarrow U\setminus(S_i\cup D_i)$, and $\rho\leftarrow\rho-w_i$.
\item \textbf{return} the pairs $(S_i,D_i)$ in order of $i$.
\end{enumerate}
\end{framed}

\noindent As in \cref{sec:cw-algorithm}, Line~5 is well defined as long as $\rho>0$; we show in \cref{lem:separator-peeling} that $\rho$ stays above $n/20$. The window of iteration $i$ is $(\rho_{i+1},\rho_i)$, where $\rho_{i+1}=\rho_i-w_i$, and we write $W_i=\Vol_{u_i}^{U_i}(r_i)$ for the volume of iteration $i$. As in Line~6, $U_{i+1}=U_i\setminus(S_i\cup D_i)$.

\subsection{Analysis}
The following three lemmas are the analogues of \cref{lem:heaviest,lem:carving-props,lem:cw-prefix}, and their proofs are the same.

\begin{lemma}[Volumes]\label{lem:separator-heaviest}
For every iteration $i$:
\begin{enumerate}[label=(\roman*),leftmargin=*]
\item $W_i\le K_i$, and $\LPVS(G[C])\le\sum_{v\in C}x_v\le\frac12W_i$ for every $C\in\cc(G[S_i])$;
\item if $i<k$, then $K_{i+1}\le\Vol_{u_i}^{U_i}(\rho_{i+1})\le W_i$, and hence $W_{i+1}\le W_i$.
\end{enumerate}
Moreover, $\sum_iW_i\le3\LPVS(G)$.
\end{lemma}
\begin{proof}
\begin{enumerate}[label=(\roman*),leftmargin=*]
\item As in \cref{lem:heaviest}, $W_i\le K_i$. Every $v\in S_i$ satisfies $r_i-d(u_i,v)+x_v>2x_v$, so it contributes its full interval length $2x_v$ to $W_i$; thus $\sum_{v\in C}2x_v\le W_i$. Restricting $(d,x)$ to $C$ is feasible for \eqref{eq:vertex-lp} on $G[C]$, which gives $\LPVS(G[C])\le\sum_{v\in C}x_v$.
\item As in the proof of \cref{lem:heaviest}, using \cref{lem:vertex-volume}(iii) and the choice of $u_i$ in Line~4.
\end{enumerate}
For the sum, a vertex $v$ contributes to $W_i$ only if $v\in U_i$ and $d(u_i,v)-x_v<r_i$, and then it lies in $S_i\cup D_i$. So every vertex contributes only to the volume of the iteration in which it is removed, and there it contributes at most $\LPVS(G)/n+2x_v$. Hence $\sum_iW_i\le\LPVS(G)+2\LPVS(G)$.
\end{proof}

\begin{lemma}[Radius]\label{lem:separator-peeling}
The radius never drops below $n/20$. \Cref{alg:vertex} terminates after at most $n$ iterations, in polynomial time, and the sets $S_i,D_i$ partition $V(G)$. Moreover, for every $i$, $|S_i|\le n/2$, and no edge joins $S_i$ to $U_{i+1}$.
\end{lemma}
\begin{proof}
As in the proof of \cref{lem:carving-props}, \cref{lem:separator-heaviest}(ii) gives $K_j\le2W_j+2(K_j-K_{j+1})$ for $j<k$. So $\sum_{j\le i}K_j\le2\sum_{j<i}W_j+2K_1\le12\LPVS(G)$, and $\rho_{i+1}=\frac n{10}-\frac{n}{240\,\LPVS(G)}\sum_{j\le i}K_j\ge\frac n{20}$.

Every iteration removes $u_i$: since $d(u_i,u_i)=0$, $r_i>0$ and $r_i\ne x_{u_i}$, the vertex $u_i$ lies in $S_i$ if $x_{u_i}<r_i$ and in $D_i$ otherwise. So there are at most $n$ iterations, and the sets $S_i,D_i$ partition $V(G)$. Since $r_i<n/10$, \cref{lem:core-separator} applied in $U_i$ gives $|S_i|\le n/2$ and shows that no edge joins $S_i$ to $U_i\setminus(S_i\cup D_i)=U_{i+1}$.
\end{proof}

\begin{lemma}[Prefix bound]\label{lem:separator-prefix}
For every $i$,
\begin{equation}\label{eq:separator-prefix}
\sum_{j\le i}|D_j|\le240\,\frac{\LPVS(G)}{n}\log\frac{3e\,\LPVS(G)}{W_i}.
\end{equation}
\end{lemma}
\begin{proof}
Let $L_j=\log\frac{3e\,\LPVS(G)}{K_j}\ge1$ be the level at the start of iteration $j$, and let $\Delta_j=\log\frac{\Vol_{u_j}^{U_j}(\rho_j)}{\Vol_{u_j}^{U_j}(\rho_{j+1})}$. As in the proof of \cref{lem:cw-prefix}, averaging over the window, with the separator size as the derivative of $\Vol_{u_j}^{U_j}$ (\cref{lem:vertex-volume}(i)), gives
\[
|D_j|\le\frac{K_j}{w_j}\min\{1,\Delta_j\}=240\,\frac{\LPVS(G)}{n}\min\{1,\Delta_j\},
\]
and \cref{lem:separator-heaviest}(ii) gives $L_{j+1}\ge L_j+\Delta_j$ for $j<k$. Summing over $j<i$, and adding at most $240\frac{\LPVS(G)}{n}\le240\frac{\LPVS(G)}{n}L_1$ for iteration $i$, we get $\sum_{j\le i}|D_j|\le240\frac{\LPVS(G)}{n}L_i\le240\frac{\LPVS(G)}{n}\log\frac{3e\,\LPVS(G)}{W_i}$, since $W_i\le K_i$.
\end{proof}

\subsection{Constructing the forest}
The final forest is computed recursively. Process connected components separately; a single vertex forms a forest of height one. On a connected graph $G$ with $n\ge2$ vertices, solve \eqref{eq:vertex-lp}, run \cref{alg:vertex}, recursively build forests for the components of each $G[S_i]$, and assemble the pairs backward: for $i$ from last to first, the forest on $U_i$ consists of the vertices of $D_i$ on a chain, with the forest already built on $U_{i+1}$ and the recursive forests for the components of $G[S_i]$ attached, as siblings, below the last vertex of the chain. If $D_i$ is empty, the forest on $U_i$ is their disjoint union. The output is the forest on $U_1=V(G)$; see \cref{fig:forest}.

\begin{figure}[t]
\centering
\begin{tikzpicture}[x=1cm,y=1cm,
  vtx/.style={circle,fill=black,inner sep=1.3pt},
  link/.style={draw=black!70},
  tri/.style={draw=blue!60!black,fill=blue!10},
  rest/.style={draw=black!60,densely dashed,fill=black!4}]
  % ---------- (a) one gluing step ----------
  \begin{scope}
    \node[vtx] (a1) at (0,0) {};
    \node[vtx] (a2) at (0,-0.55) {};
    \node[vtx] (a3) at (0,-1.1) {};
    \draw[link] (a1)--(a2)--(a3);
    \draw[decorate,decoration={brace,amplitude=4pt}] (-0.25,-1.2) -- (-0.25,0.1) node[midway,left=5pt,font=\footnotesize] {$D_i$};
    \foreach \x in {-1.95,-1.05} {
      \draw[link] (a3) -- (\x,-1.85);
      \draw[tri] (\x,-1.85) -- ++(-0.3,-0.75) -- ++(0.6,0) -- cycle;
    }
    \node[font=\footnotesize,align=center] at (-1.5,-3.05) {forests of the\\components of $G[S_i]$};
    \draw[link] (a3) -- (1.3,-1.85);
    \draw[rest] (1.3,-1.85) -- ++(-0.55,-0.75) -- ++(1.1,0) -- cycle;
    \node[font=\footnotesize,align=center] at (1.3,-3.05) {forest on\\$U_{i+1}$};
    \node[font=\footnotesize] at (-0.1,-4.7) {(a) the forest on $U_i$};
  \end{scope}
  % ---------- (b) the whole forest ----------
  \begin{scope}[xshift=6.0cm]
    \draw[red!35,line width=4pt,line cap=round,line join=round] (0,0)--(0,-0.5)--(1.2,-1.1)--(1.2,-1.6)--(1.15,-2.2)--(1.15,-2.75);
    \draw[rest,rounded corners=3pt,fill=none] (0.15,-0.85) rectangle (3.05,-4.0);
    \node[font=\scriptsize,anchor=north east] at (3.05,-4.02) {forest on $U_2$};
    % level 1
    \node[vtx] (b1) at (0,0) {};
    \node[vtx] (b2) at (0,-0.5) {};
    \draw[link] (b1)--(b2);
    \draw[decorate,decoration={brace,amplitude=3pt}] (-0.18,-0.58) -- (-0.18,0.08) node[midway,left=4pt,font=\footnotesize] {$D_1$};
    \foreach \x in {-1.0,-0.35} {
      \draw[link] (b2) -- (\x,-1.1);
      \draw[tri] (\x,-1.1) -- ++(-0.25,-0.6) -- ++(0.5,0) -- cycle;
    }
    \node[font=\scriptsize] at (-0.68,-1.95) {$G[S_1]$};
    % level 2
    \node[vtx] (c1) at (1.2,-1.1) {};
    \node[vtx] (c2) at (1.2,-1.6) {};
    \draw[link] (b2)--(c1)--(c2);
    \draw[decorate,decoration={brace,amplitude=3pt}] (1.38,-1.02) -- (1.38,-1.68) node[midway,right=4pt,font=\footnotesize] {$D_2$};
    \foreach \x in {0.5,1.15} {
      \draw[link] (c2) -- (\x,-2.2);
      \draw[tri] (\x,-2.2) -- ++(-0.25,-0.6) -- ++(0.5,0) -- cycle;
    }
    \node[font=\scriptsize] at (0.82,-3.05) {$G[S_2]$};
    % last level
    \draw[link,densely dotted] (c2) -- (2.1,-2.2);
    \node[vtx] (e1) at (2.1,-2.2) {};
    \node[vtx] (e2) at (2.1,-2.7) {};
    \draw[link] (e1)--(e2);
    \draw[decorate,decoration={brace,amplitude=3pt}] (2.28,-2.12) -- (2.28,-2.78) node[midway,right=4pt,font=\footnotesize] {$D_k$};
    \foreach \x in {1.8,2.45} {
      \draw[link] (e2) -- (\x,-3.0);
      \draw[tri] (\x,-3.0) -- ++(-0.22,-0.5) -- ++(0.44,0) -- cycle;
    }
    \node[font=\scriptsize] at (2.12,-3.75) {$G[S_k]$};
    \node[font=\footnotesize] at (0.9,-4.7) {(b) the forest on $V(G)=U_1$};
  \end{scope}
\end{tikzpicture}
\caption{Assembling the elimination forest. (a)~The vertices of the separator $D_i$ form a chain; below its last vertex hang, as siblings, the recursive forests of the components of $G[S_i]$ and the forest already built on $U_{i+1}$. (b)~Applying (a) for $i=k,k-1,\ldots,1$: the dashed box is the forest on $U_2$, glued below the last vertex of $D_1$, and the dotted edge stands for the omitted levels $D_3,\ldots,D_{k-1}$. A root-to-leaf path passes through the chains $D_1,\ldots,D_i$ and then ends on $D_i$ or enters a component of $G[S_i]$ (highlighted), which gives the height bound \eqref{eq:forest-recurrence}.}
\label{fig:forest}
\end{figure}
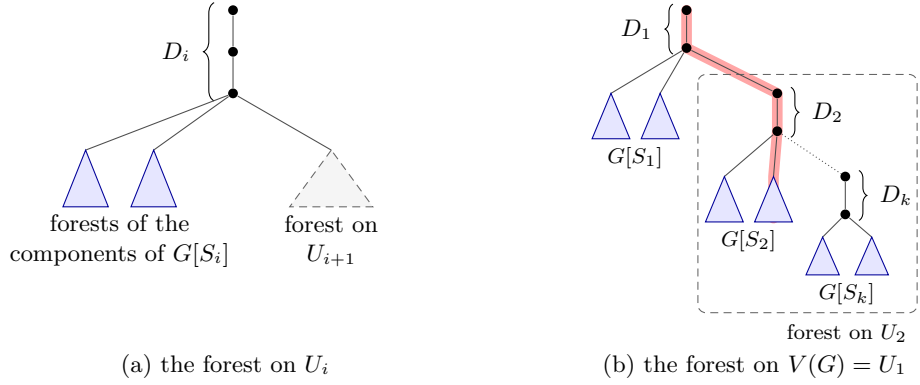

\begin{lemma}[Assembly]\label{lem:forest-assembly}
The construction gives an elimination forest of $G$ of height
\begin{equation}\label{eq:forest-recurrence}
h(G)\le\max_i\Big\{\sum_{j\le i}|D_j|
+\max_{C\in\cc(G[S_i])}h(G[C])\Big\},
\end{equation}
where the inner maximum is zero when $S_i$ is empty.
\end{lemma}
\begin{proof}
We show by backward induction on $i$ that the forest on $U_i$ is an elimination forest of $G[U_i]$; the forest on $U_{i+1}$ is one of $G[U_{i+1}]$ by induction, and each recursive forest is an elimination forest of its component by induction on the size. Consider an edge $pq$ of $G[U_i]$. If $p\in D_i$, then $p$ lies on the chain, and every vertex of $U_i$, in particular $q$, lies on the chain or below its last vertex, so $p$ and $q$ are comparable; the case $q\in D_i$ is symmetric. If $p$ and $q$ lie in the same component of $G[S_i]$ or both lie in $U_{i+1}$, they are comparable by induction. The remaining cases do not occur: distinct components of $G[S_i]$ are not adjacent, and no edge joins $S_i$ to $U_{i+1}$ by \cref{lem:separator-peeling}.

For the height, a root-to-leaf path of the forest on $V(G)$ passes through the chains $D_1,\ldots,D_{i-1}$ and then either ends on the chain $D_i$ or enters a component of $G[S_i]$, for some $i$. In the first case it has at most $\sum_{j\le i}|D_j|$ vertices; in the second case, at most $\sum_{j\le i}|D_j|+h(G[C])$ vertices for the component $C$ that it enters. This proves \eqref{eq:forest-recurrence}.
\end{proof}

\begin{lemma}[Height]\label{lem:forest-height}
For every connected graph $G$ with $n$ vertices, the construction gives an elimination forest of height
\begin{equation}\label{eq:forest-pw-bound}
h(G)=O(\pw(G)\log n).
\end{equation}
\end{lemma}
\begin{proof}
Since $G$ has an edge, $\pw(G)\ge1$. Unrolling \eqref{eq:forest-recurrence}, there are instances $H_0=G,H_1,\ldots,H_T$, each with at least two vertices, and indices $i_0,\ldots,i_T$ such that $H_{t+1}$ is a component of $H_t[S_{i_t}]$ for $t<T$ and
\[
h(G)\le1+\sum_{t=0}^T\,\sum_{j\le i_t}|D_j|,
\]
where the inner sum is taken in $H_t$, and the term $1$ bounds the height of a forest on a single vertex. In what follows, $W_{i_t}$ also refers to the run of \cref{alg:vertex} on $H_t$. By \cref{lem:vertex-lower}, $\LPVS(H_t)/|V(H_t)|\le\pw(H_t)\le\pw(G)$, since pathwidth does not increase under taking induced subgraphs. So by \cref{lem:separator-prefix},
\[
\sum_{j\le i_t}|D_j|\le240\,\pw(G)\log\frac{3e\,\LPVS(H_t)}{W_{i_t}}.
\]
By \cref{lem:separator-heaviest}(i), $\LPVS(H_{t+1})\le W_{i_t}$ for $t<T$, and by \cref{lem:vertex-volume}(ii), $W_{i_T}\ge\LPVS(H_T)/|V(H_T)|$. Therefore
\[
\sum_{t=0}^T\log\frac{3e\,\LPVS(H_t)}{W_{i_t}}
\le(T+1)\log(3e)+\log\frac{\LPVS(H_0)}{\LPVS(H_T)}+\log|V(H_T)|.
\]
By \cref{lem:separator-peeling}, $|V(H_{t+1})|\le|V(H_t)|/2$, so $T=O(\log n)$. By \cref{lem:vertex-lower}, $\LPVS(G)\le n\pw(G)\le n^2$. Finally, every connected graph $H'$ with $s\ge2$ vertices has $\LPVS(H')\ge3/16$: for two vertices $u,v$ joined by a simple path, the edge constraints and the triangle inequality give $d(u,v)\le\sum_{pq}(x_p+x_q)\le2\LPVS(H')$, the sum ranging over the edges of the path, while spreading with $S=V(H')$ gives $\sum_vd(u,v)\ge(s^2-1)/4$; so $\LPVS(H')\ge(s^2-1)/(8s)\ge3/16$. Hence the sum above is $O(\log n)$, and $h(G)=O(\pw(G)\log n)$.
\end{proof}
All steps take polynomial time, by \cref{lem:separator-peeling} and since the recursion has at most $n$ instances with at least two vertices at each depth.

\begin{proof}[Proof of \cref{thm:main-treedepth}]
The construction processes connected components separately. Apply \eqref{eq:forest-pw-bound} to every component with an edge, using that pathwidth does not increase under taking subgraphs, and note that single vertices have height one; then combine with $\pw(G)\le\td(G)$.
\end{proof}

\begin{proof}[Proof of the pathwidth statement in \cref{thm:main-layout}]
Take a depth-first preorder of the forest. By \cref{lem:forest-order} and \eqref{eq:forest-pw-bound}, applied to every component with an edge, its vertex-separation cost is $O(\pw(G)\log n)$, and \cref{lem:vs-pw} turns it into a path decomposition of the same width.
\end{proof}

\section{Conclusion and Open Problems}\label{sec:conclusion}
We introduced Greedy Region Growing, with two algorithms: Greedy Seymour Carving, a greedy version of Seymour recursion that controls cumulative boundary costs, and Global Carving, which pays for every cut with the growth of a global level rather than with its own piece. Together, they give $O(\log n)$-approximations for \textsc{Cutwidth}, \textsc{Pathwidth}, \textsc{Treedepth}, \textsc{Minimum Storage}, and \textsc{Minimum Register Sufficiency}, and a flow-expander decomposition with $O(\log n\log\log n)$ overhead.

Our algorithms for the layout problems, like those of~\cite{BKS24}, are LP-based, and the standard integrality gaps on expanders~\cite{BKS24} show that $O(\log n)$ is the limit of these relaxations for all of them (\cref{app:gaps}). On the hardness side, under the Small Set Expansion Hypothesis, several of these problems admit no constant-factor approximation~\cite{WAPL14}, and for \textsc{Treedepth} the known NP-hardness of approximation is only for a factor of $1.0003$~\cite{BNS25}. Can SDP relaxations give better bounds for these problems?

For flow-expander decompositions, the $\Omega(\log n)$ graph-theoretic lower bound on the overhead shows that the logarithmic loss is unavoidable, so only the factor $O(\log\log n)$ remains. Global Carving does not remove it: the fine-grained flow-cut gap (\cref{thm:fine-grained}) needs a bound on the sum of the costs of all pieces, each relative to its own weight, whereas Global Carving controls only the prefixes of the carving. We conjecture that the factor $O(\log\log n)$ can be removed from the fine-grained flow-cut gap, which would give a flow-expander decomposition with optimal $O(\log n)$ overhead.

\section*{Acknowledgements}
The author thanks Nikhil Bansal, Ian DeHaan, Euiwoong Lee, Yury Makarychev and Thatchaphol Saranurak for useful discussions.
Generative AI tools assisted in developing the algorithms and proofs and preparing the manuscript; details appear in \cref{app:ai}.

\bibliographystyle{alpha}
\bibliography{references}

\appendix
\section{Comparison with the Weak Local Partitions of~\cite{BKS24}}\label{app:bks-wlp}
The weak-local-partition condition in~\cite{BKS24} compares parts by volume without fixing their order: for every part, the sum ranges over all parts of at least the same volume. Our ordered condition \eqref{eq:wlp} implies it. Indeed, given a part $X_i$, let $t$ be the largest index with $\Vol(X_t)\ge\Vol(X_i)$; then all parts of volume at least $\Vol(X_i)$ are among $X_1,\ldots,X_t$, and \eqref{eq:wlp} for $t$ gives the bound, since $\Vol(X_t)\ge\Vol(X_i)$.

The presentation of~\cite{BKS24} also notes that the technical volume includes partial boundary-edge contributions and a seed. Our condition uses the volume \eqref{eq:static-volume}. We show here that the carving of \cref{alg:cutwidth} also satisfies the unordered bound with the carving-time volumes $W_i$, which include both of those contributions. By \cref{lem:heaviest}, $W_1\ge W_2\ge\cdots$. Fix $i$, and let $t$ be the largest index with $W_t\ge W_i$. Then $t\ge i$, so $W_t=W_i$, and every $j$ with $W_j\ge W_i$ satisfies $j\le t$. By \cref{lem:cw-prefix},
\[
\sum_{j:\,W_j\ge W_i}|\delta_G(X_j)|
\le\sum_{j\le t}|\delta_G(X_j)|
\le320\,\frac{\LP(G)}{n}\log\frac{2e\,\LP(G)}{W_i}.
\]

\section{Minimum Register Sufficiency}\label{app:register}
We reduce \textsc{Minimum Register Sufficiency} to weighted \textsc{Minimum Storage}. For a topological ordering $\pi$ of a DAG $G=(V,E)$, let
\[
R_i^\pi=\set{u:\pi(u)\le i,\ uv\in E\text{ for some }\pi(v)>i}.
\]
The register sufficiency of $\pi$ is $\max_i|R_i^\pi|$, and $\MRS(G)$ is the minimum over topological orderings.

Let $G=(V,E)$ have $m$ arcs, and let $\eps=1/(10m)$. We make two copies $v^-$ and $v^+$ of every vertex $v$, and join them by an arc $v^-v^+$ of weight $1$. For every arc $uv\in E$, we add the arcs $u^-v^-$ and $v^-u^+$ of weight $\eps$. Call the resulting weighted graph $G'$. The copies $v^-$ induce a copy of $G$, and the copies $v^+$ are sinks, so $G'$ is acyclic.

The idea is the following. In any topological ordering of $G'$, the vertex $u^+$ can appear only after $v^-$ for all out-neighbors $v$ of $u$. So if the copies $v^-$ appear in the order $\pi$, then for every $u\in R_i^\pi$ the arc $u^-u^+$ crosses the prefix ending just before the $(i+1)$-th copy $v^-$. Conversely, if every $u^+$ is placed as early as possible, then essentially no other arcs of weight $1$ cross this prefix. The $\eps$-arcs have total weight $2m\eps=1/5$, so they do not affect the analysis.

\begin{lemma}\label{lem:register-reduction}
The graph $G'$ has $2n$ vertices and arc weights in $[1/(10m),1]$, and
\begin{equation}\label{eq:register-reduction}
\MRS(G)\le\MS_w(G')\le3\MRS(G).
\end{equation}
Moreover, the restriction of any topological ordering $\sigma$ of $G'$ to $\{v^-:v\in V\}$ is a topological ordering of $G$ of register sufficiency at most $\MS^w_\sigma(G')$.
\end{lemma}
\begin{proof}
\emph{Upper bound.} Let $\pi=(v_1,\ldots,v_n)$ be a topological ordering of $G$ with register sufficiency $P=\MRS(G)\ge1$. Let $\sigma$ list $v_1^-,\ldots,v_n^-$ in this order, and insert each $u^+$ directly after the first $v_i^-$ with $u\in V\setminus R_i^\pi$ and $\pi(u)\le i$. Then $\sigma$ is topological. A prefix $Z$ of $\sigma$ that ends between $v_i^-$ and $v_{i+1}^-$ is crossed by the unit arcs $u^-u^+$ with $u\in R_{i-1}^\pi\cup\{v_i\}$ only. Indeed, if $u^-\in Z$ and $u^+\notin Z$, then either $u=v_i$, or $\pi(u)<i$ and, by the placement of $u^+$, $u\in R_j^\pi$ for all $\pi(u)\le j<i$, in particular $u\in R_{i-1}^\pi$. Hence $w(\delta_{G'}(Z))\le P+1+1/5\le3P$. Hence $\MS_w(G')\le\MS^w_\sigma(G')\le3P$.

\emph{Lower bound and restriction.} Let $\sigma$ be a topological ordering of $G'$, and let $\pi$ be its restriction to $\{v^-:v\in V\}$. The arcs $u^-v^-$ make $\pi$ topological in $G$; write $\pi=(v_1,\ldots,v_n)$. Fix $i<n$, let $Z$ be the prefix of $\sigma$ ending just before $v_{i+1}^-$, and let $u\in R_i^\pi$ with an arc $uv$, $\pi(v)>i$. Then $u^-\in Z$ and $v^-\notin Z$, so $u^+\notin Z$ by the arc $v^-u^+$. Thus $u^-u^+\in\delta_{G'}(Z)$ for every $u\in R_i^\pi$, and $|R_i^\pi|\le w(\delta_{G'}(Z))\le\MS^w_\sigma(G')$. With $\sigma$ optimal, this also gives $\MRS(G)\le\MS_w(G')$.
\end{proof}

\begin{proof}[Proof of the Minimum Register Sufficiency statement in \cref{thm:main-directed}]
Apply \cref{rem:weighted-storage} to the graph $G'$ of \cref{lem:register-reduction}, which has $2n$ vertices and weight ratio at most $10n^2$, so the approximation factor is $O(\log n)$. Restricting the output order to the copies $v^-$ gives the same approximation factor for $\MRS(G)$.
\end{proof}

\section{Nonnegative Node-Weightings}\label{app:nonnegative}
The assumption $A\ge1$ in \cref{thm:weighted-expander} is used only to bound the deficits in \cref{lem:expander-volume}(ii), to ensure $A(X_i)\ge1$ in \cref{lem:expander-well-defined}, and in the last step of the charging argument in the proof of \cref{thm:weighted-expander}. It can be removed as follows.

\begin{lemma}\label{lem:nonnegative-weights}
Let $G$ be a graph with $n$ vertices, let $A\ge0$ be a rational node-weighting, and let $\phi>0$ be rational. One can delete in polynomial time $O(\phi|A|\log n\log\log n)$ edges so that $A$ is $\phi$-flow-expanding in every remaining connected component.
\end{lemma}
\begin{proof}
Let $A'=\frac{n}{|A|}A+1$ and $\phi'=\frac{|A|}{n}\phi$, so that $A'\ge1$ and $|A'|=2n$. Apply \cref{thm:weighted-expander} to $A'$ and $\phi'$; it deletes $O(\phi'|A'|\log|A'|\log\log|A'|)=O(\phi|A|\log n\log\log n)$ edges. Let $H$ be a remaining component, and let $D$ be an $A$-respecting demand in $H$. Then $\frac{n}{|A|}D$ is $A'$-respecting, so it can be routed in $H$ with congestion at most $1/\phi'$. Scaling this routing by $\frac{|A|}{n}$ routes $D$ with congestion at most $\frac{|A|}{n}\cdot\frac1{\phi'}=\frac1\phi$.
\end{proof}

\section{Separation Oracles and Implementation}\label{app:implementation}
We give the algorithmic details for the LPs, the region-growing steps, and the multicommodity flow LP of \cref{sec:expander}.

\subsection{Separating the spreading constraints}
All triangle, symmetry, nonnegativity, precedence-zero, and edge constraints in \eqref{eq:cw-lp}, \eqref{eq:directed-lp}, and \eqref{eq:vertex-lp} form polynomial-size families. Only the spreading constraints are exponential.

Fix a candidate solution and a center $u$. For either symmetric relaxation, sort the distances as $q_1\le\cdots\le q_n$. Among all sets of size $s$, the minimum of $\sum_{v\in S}d(u,v)$ is $\sum_{i=1}^sq_i$. A spreading constraint centered at $u$ is therefore violated if and only if some prefix satisfies
\[
\sum_{i=1}^sq_i<\frac{s^2-1}{4}.
\]
Check every prefix after sorting, and repeat for every center. For the directed relaxation, sort $q_v=d(u,v)+d(v,u)$ and apply the same procedure. The vertex variables in \eqref{eq:vertex-lp} do not affect this separation oracle.

For completeness, the variables may be bounded by $n^2$ without changing the optimum. Capping a nonnegative metric or quasi-metric at $n^2$ preserves the triangle inequalities. Any spreading sum containing a capped term is already at least its required right-hand side; all other sums are unchanged. In the vertex LP, cap the $x_v$ variables at $n^2$ as well. If an endpoint variable is capped, it still dominates every capped edge distance, so the edge constraints remain valid. All coefficients have polynomial encoding length, and the ellipsoid method computes an optimum rational solution in polynomial time.

\subsection{Implementing the carving steps}
Each LP solution is rational with polynomial encoding length, so all volumes, weights, event radii, keys, and window endpoints are rational, and so is every quantity that the algorithms compare. Global Carving (\cref{alg:cutwidth,alg:directed,alg:vertex}) only compares volumes. In Line~5, the boundary is constant between consecutive events, so it suffices to take the midpoint of an event-free interval in the window with the smallest boundary cost.

In \cref{sec:expander}, logarithms enter the algorithm only through the stability tests and through $\Gamma$ in Line~7. A stability test $F(a)\le2F(b)$ for a deficit $F(r)=\log(e\Lambda/\mathcal{V}(r))$ is equivalent to $\mathcal{V}(b)^2\le e\Lambda\,\mathcal{V}(a)$, and the test for $\Gamma$ is equivalent to $A(B_b)^2\le e|A|\,A(B_a)$ for the two balls $B_a\subseteq B_b$. Replacing the constant $e$ by $4$ inside every deficit and inside $\Gamma$ makes these rational comparisons. The new deficits are still at least one and exceed the old ones by less than $\frac12$, so the analysis changes only by absolute constants. The annulus count only needs to be an integer of order $\log\log|A|$ that is large enough for \cref{lem:expander-well-defined}; such an integer can be computed by comparisons with powers of two.

In Line~7, it suffices to use $1+\lceil\log_2x\rceil$ for $x=|A|/A(X_i)$ in place of $\Gamma(X_i)$. For $x\ge1$, both quantities lie between $\log(ex)$ and $3\log(ex)$, so the averaging argument in the proof of \cref{lem:size-cut} loses only a constant factor.

Finally, the radius of \cref{lem:stable-annulus} can be taken as the midpoint $m$ of an event-free interval inside the annulus that minimizes $c(m)/\mathcal{V}(m)$. On each such interval $\mathcal{V}$ is affine and nondecreasing, so $\mathcal{V}(m)\ge\frac12\mathcal{V}(r)$ for every $r$ in it, and the minimizing midpoint has $c(m)/\mathcal{V}(m)$ at most twice the value at the radius given by \cref{lem:rg}. Since $F(m)\ge F(b)$, the proof of \cref{lem:stable-annulus} shows that $m$ satisfies its bound up to a factor of two. All carving estimates allow these constant factors, and all algorithms run in polynomial time.

\subsection{The multicommodity flow LP}
There is one commodity $p=\{s,t\}$ for each unordered vertex pair. Replace each undirected edge by two orientations, and let $f_{xy}^p\ge0$ be the amount of commodity $p$ sent on orientation $xy$. Flow conservation requires net outflow $\Pi(s,t)$ at $s$, net inflow $\Pi(s,t)$ at $t$, and zero net flow elsewhere. Minimize $\lambda\ge0$ subject to
\[
\sum_p(f_{xy}^p+f_{yx}^p)\le\lambda
\qquad(xy\in E(G)).
\]
This LP has polynomial size. Its length dual, normalized by $\sum_e\ell_e\le1$, maximizes
\[
\sum_{\{s,t\}}\Pi(s,t)d_\ell(s,t).
\]
By strong duality, a routing of congestion at most $1/(2\phi)$ exists exactly when every normalized length vector has dual value at most $1/(2\phi)$. So either the LP returns the routing required by \cref{lem:product-certificate}, or its dual returns lengths satisfying \eqref{eq:failed-dual}.

\section{Integrality Gaps}\label{app:gaps}
We sketch why $O(\log n)$ is the limit of the relaxations we use for the layout problems. Here the \emph{integrality gap} for \textsc{Cutwidth} is the ratio of $\cw(G)$ to
\[
\max_{\emptyset\ne S\subseteq V(G)}\frac{\LP(G[S])}{|S|},
\]
the strongest normalized LP lower bound over all nonempty induced subgraphs, which captures the LP bounds used throughout the recursion. We use the analogous bound for the other layout problems, with the corresponding LP (after the reduction for \textsc{Minimum Register Sufficiency}). For \textsc{Cutwidth} and \textsc{Pathwidth}, the gap is witnessed by the standard example of bounded-degree expanders~\cite{BKS24}; the same example works for the other problems. Let $H$ be an $n$-vertex expander of maximum degree $O(1)$, in which every set of $\lceil n/2\rceil$ vertices has $\Omega(n)$ edges leaving it, and let $s=Cn/\log n$ for a large constant $C$. Write $\operatorname{hop}(u,v)$ for the number of edges on a shortest path in the underlying undirected graph.

\medskip
\noindent\emph{Spreading.} Let $F$ be a graph with $N=O(n)$ vertices and maximum degree $O(1)$, and let $\rho$ be a symmetric function with $\rho(u,v)\ge s\max\{1,c\operatorname{hop}(u,v)\}$ for all $u\ne v$, where $c>0$ is a constant. Then $\rho$ satisfies the spreading constraints. Indeed, fix $u$ and $S$ with $|S|=k\ge2$. If $k\le4s-1$, then $\sum_{v\in S}\rho(u,v)\ge s(k-1)\ge(k^2-1)/4$. Otherwise, since a ball of $r$ hops contains $2^{O(r)}$ vertices, at least half of $S$ is $\Omega(\log k)=\Omega(\log n)$ hops away from $u$, so the sum is $\Omega(sk\log n)=\Omega(Cnk)\ge k^2/4$.

\medskip
\noindent\emph{Undirected problems.} Let $d=s\operatorname{hop}_H$, and $x_v=s/2$ for every $v$. This is feasible for \eqref{eq:cw-lp} and \eqref{eq:vertex-lp}, with objective $O(sn)=O(n^2/\log n)$. On the other hand, in every ordering, the prefix of $\lceil n/2\rceil$ vertices has $\Omega(n)$ edges, and hence $\Omega(n)$ vertices with a neighbor, outside it. So $\cw(H)$, $\pw(H)$, and $\td(H)>\pw(H)$ are $\Omega(n)$, which is $\Omega(\log n)$ times $\LP(H)/n$ and $\LPVS(H)/n$.

\medskip
\noindent\emph{Directed problems.} Let $G$ be the DAG on $V(H)\cup E(H)$ with an arc $ue$ whenever $u\in e$. Every topological ordering has a prefix containing exactly $\lceil n/2\rceil$ vertices of $V(H)$. Every edge $e=uv$ of $H$ with $u$ in this prefix and $v$ outside lies outside it, so the arc $ue$ crosses the prefix and $u$ has an out-neighbor outside it. Hence $\MS(G)$ and $\MRS(G)$ are $\Omega(n)$. Now let $d(x,y)$ be $s$ times the minimum number of arcs traversed forward on an undirected $x$--$y$ path. It satisfies the triangle inequality and $d(y,x)=0$ for every arc $xy$, and its objective in \eqref{eq:directed-lp} is $s|E(G)|=O(n^2/\log n)$. Since every arc goes from $V(H)$ to $E(H)$, forward and backward arcs alternate along every path, and since $G$ is acyclic, $\bar d(x,y)\ge s$ for $x\ne y$. Hence $\bar d(x,y)\ge s\max\{1,\operatorname{hop}(x,y)-1\}$, and spreading holds. For \textsc{Minimum Register Sufficiency}, the relaxation is \eqref{eq:directed-lp} on the graph obtained from $G$ by \cref{lem:register-reduction} and \cref{rem:weighted-storage}. The same construction, with $s=CN/\log N$ for its number $N=\Theta(n^2)$ of vertices, gives an $\Omega(\log n)$ gap there: that graph has no directed path with more than four arcs, so forward arcs form a constant fraction of every long path, and it has $O(n)2^{O(r)}$ vertices within $r$ hops of any vertex, which suffices for the computation above.

Finally, these bounds hold uniformly over induced subgraphs. The displayed feasible solutions restrict to every induced subgraph $J$, with objective $O(s|V(J)|)$: the original constructions have bounded degree, the vertex solution has $x_v=s/2$, and every arc of the subdivided construction is incident to a subdivision vertex of degree two. Thus maximizing the normalized LP value over induced subgraphs preserves the $\Omega(\log n)$ gaps.

\paragraph{The direct cutwidth LP.}
Let $H$ be the same $n$-vertex bounded-degree expander as above, so that $\cw(H)=\Omega(n)$.
We also show that the stronger direct relaxation $\mathrm{LPCW}$ introduced in~\cite[Section~6.2]{BKS24} has an $\Omega(\log n)$ integrality gap. This LP minimizes $C$ over metrics $d_x$ on $V(H)$, one for each $x$, subject to $\sum_{uv\in E(H)}d_x(u,v)\le C$ for every $x$, the pair constraints $d_x(x,y)+d_y(x,y)\ge1$ for distinct $x,y$, and the triple constraints $d_x(y,z)+d_y(x,z)+d_z(x,y)\ge1$ for distinct $x,y,z$.

Write $h(u,v)=\operatorname{hop}_H(u,v)$. For each $x$, give every edge $uv$ length
\[
\ell_x(uv)=\frac{2}{1+\min\{h(x,u),h(x,v)\}},
\]
and let $d_x$ be the resulting shortest-path metric. The metric constraints hold automatically. Every path from $x$ to a distinct vertex starts with an edge of length $2$, so the pair constraints hold. For a triple, relabel its vertices so that $y,z$ are a farthest pair in hop distance, and put $a=h(y,z)$. Any $y$--$z$ path has at least $a$ edges. For $0\le j<a$, its $j$-th vertex $v_j$, counting from $v_0=y$, satisfies $h(x,v_j)\le h(x,y)+j\le a+j\le2a-1$. Thus each of the first $a$ edges has length at least $1/a$, giving $d_x(y,z)\ge1$ and hence the triple constraint.

Let $\Delta\ge3$ be a fixed upper bound on the maximum degree of $H$, and put $r=\lfloor\log n/(2\log\Delta)\rfloor$. Only $O(\Delta^r)=O(\sqrt n)$ edges have an endpoint within $r$ hops of $x$; each has length at most $2$, and all other edges have length at most $2/(1+r)$. Therefore, uniformly over $x$,
\[
\sum_{uv\in E(H)}d_x(u,v)
\le\sum_{uv\in E(H)}\ell_x(uv)
=O\left(\sqrt n+\frac{n}{\log n}\right)
=O\left(\frac{n}{\log n}\right).
\]
This gives a feasible direct LP solution of value $O(n/\log n)$, while $\cw(H)=\Omega(n)$, proving the claimed gap.

\section{AI Disclosure}\label{app:ai}
The author formulated the fine-grained flow-cut gap problem, specified the LP relaxation, and proposed greedy Seymour recursion as the approach. A central difficulty was choosing a greedy order when the radii returned by Seymour's region growing were not known in advance. Drawing on the stability idea in~\cite{BKS24}, the author proposed maintaining a stable interval $[r,r']$ for each vertex and selecting centers greedily by ball volume at an interval endpoint, rather than at the eventual carving radius. The author also anticipated that using stability within greedy Seymour recursion would avoid the loss from the multiscale packing-covering construction of~\cite{BKS24} and improve its bound. The author supplied this algorithmic template and explained the relevant Seymour bound. Through iterative discussions, GPT-5.6 Sol developed the proof establishing the desired guarantee (\cref{thm:fine-grained}). It subsequently extended the result to the other problems, leading to $O(\log n\log\log n)$ bounds for all of them.

The improvement to $O(\log n)$ for the layout problems (Global Carving, \cref{sec:ov-paced,sec:cutwidth,sec:directed,sec:vertex}) was obtained by Claude Opus 5.5, when given the paper with these results. GPT-6 Astra, Sol and Opus proved the integrality gaps and assisted with writing, and together drew the figures. All results were verified by the author.
\end{document}